\documentclass[journal]{IEEEtran}
\ifCLASSINFOpdf
\else
\fi
\usepackage{graphicx}
\usepackage{amssymb}
\usepackage{amsmath}
\usepackage{multicol}
\usepackage{stfloats}
\usepackage{cite}
\usepackage{listings}
\usepackage{booktabs}
\usepackage{makecell}
\newtheorem{theorem}{Theorem}
\newtheorem{proposition}{Proposition}
\newtheorem{corollary}{Corollary}
\newtheorem{lemma}{Lemma}
\newtheorem{definition}{Definition}
\newtheorem{assumption}{Assumption}
\newtheorem{remark}{Remark}
\newtheorem{proof}{Proof}
\newtheorem{problem}{Problem}

\usepackage{tikz,xcolor,hyperref}
\definecolor{lime}{HTML}{A6CE39}
\DeclareRobustCommand{\orcidicon}{
	\begin{tikzpicture}
		\draw[lime, fill=lime] (0,0) 
		circle [radius=0.16] 
		node[white] {{\fontfamily{qag}\selectfont \tiny ID}}; 
		\draw[white, fill=white] (-0.0625,0.095) 
		circle [radius=0.007];	  
	\end{tikzpicture}
	\hspace{-2mm}}
\foreach \x in {A, ..., Z}{
	\expandafter\xdef\csname orcid\x\endcsname{\noexpand\href{https://orcid.org/\csname orcidauthor\x\endcsname}{\noexpand\orcidicon}}
}

\makeatother

\begin{document}
\title{A Distributed Cluster Economic Dispatch Scheme for Cross-regional Microgrids Induced by Well-designed Communication Weights}
\author{Yalin Zhang\orcidA{},~\IEEEmembership{Member,~IEEE,}
	Zhongxin Liu\orcidB{},~\IEEEmembership{Member,~IEEE,}
	Yulin Chen\orcidC{},~\IEEEmembership{Member,~IEEE,}
	Donglian Qi\orcidE{},~\IEEEmembership{Senior Member,~IEEE,}
	and Zengqiang Chen\orcidD{}
	\thanks{Manuscript received XX, XX; revised XX, XX;
		accepted XX, XX. This work is supported in part by the Project of Sanya Yazhou Bay Science and Technology City (Grant No. SKJC-JYRC-2024-66) and in part by the National Natural Science Foundation of China (Grant No. 52477133) (Corresponding author: Zhongxin Liu and Yulin Chen.). 
		\par Yalin Zhang and Yulin Chen are with the Hainan Institute, Zhejiang University, Sanya 572025, China. Zhongxin Liu and Zengqiang Chen are with the College of Artificial Intelligence, Nankai University, Tianjin 300350, and also with the Tianjin Key Laboratory of Interventional Brain-Computer Interface and Intelligent Rehabilitation, Nankai University, Tianjin 300350, China. Donglian Qi is with the College of of Electrical Engineering, Zhejiang University, Hangzhou 310007, China (e-mail: zhangyalin@zju.edu.cn; lzhx@nankai.edu.cn; chenyl2017@zju.edu.cn; qidl@zju.edu.cn; chenzq@nankai.edu.cn).
}}
\markboth{IEEE/CAA JOURNAL OF AUTOMATICA SINICA, VOL. XX, NO. XX, XX XX}
{Zhang \MakeLowercase{\textit{et al.}}: Distributed Cluster Economic Dispatch Scheme of Cross-regional Microgrids Aggregated by A VPP Induced by Well-designed Communication Weights}
\maketitle
\begin{abstract}
	A large-scale microgrid typically consists of several cross-regional subgrids aggregated by a virtual power plant (VPP). However, current consensus based schemes can-not guarantee the feature of differential demand between subgrids. Thus, distributed cluster consensus control induced by communication weights is investigated in this paper to solve the ED problem of a large-scale microgrid, which can achieve the expected cluster via well-designed communication weights. A communication weight matrix design method for a directed and connected graph based on eigenvector centrality is designed, which enables the adjacency matrix of the communication network to have a given leading eigenvector and allows agents in each cluster to have the same eigenvector center value. Based on this, a distributed cluster ED scheme, namely a leader-follower cluster consensus controller, is designed to drive marginal cost (MC) to achieve multiconsensus, thus allocating power among DGs. In addition, the power deficit of each subgrid collected by a VPP can be allocated to utility grids according to predetermined ratios, thus maintaining power supply-demand balance of each subgrid. For this scheme, it should be emphasized that the weighted network used is directed and connected; meanwhile, leader information only can be accessed by a few clusters. Correspondingly, relevant simulations are attached to verify the effectiveness of the designed scheme.
\end{abstract}
\IEEEpeerreviewmaketitle
\begin{center}
	N\footnotesize{OMENCLATURE}
\end{center}
\normalsize
\begin{tabbing}
	\hspace{2cm} \= \kill
	BESS\> battery energy storage system\\
	DG\> distributed generation\\
	VPP\>virtual power plant\\
	UG\>utility grid\\
	ED\>economic dispatch\\
	MG\> microgrid\\
	NMGs\>networked MGs\\
	MAS\> multi-agent system\\
	LFCC\>leader-follower cluster consensus\\
	$\mathrm{m}$, $\mu$\> the numbers of subgrids and UGs, respectively\\
	$\mathrm{n}_j-\mathrm{n}_{j-1}$\> the number of DGs in subgrid $j$ \\
	$\alpha_i$, $\beta_i$\>the cost function coefficients of DG $i$\\
	$P_i$, $C_i(P_i)$\>the output power and cost function of DG $i$\\
	$C^j(P_{\mathrm{n}_{j-1}+1},\cdots,P_{\mathrm{n}_j},P_{\mathrm{M},j})$\>\\
	\>the total electricity cost of subgrid $j$\\
	$P_{\mathrm{M},j}$\> the total exchange power between\\
	\>subgrid $j$ and UGs\\
	$\rho_k$, $\rho$\>the electric price of UG $k$ and its stack vector\\
	$\lambda_i$\>the MC of DG $i$\\
	$W$\>the aggregation weight\\
	$\omega_{jk}$\> the total exchange power ratio between\\
	\>subgrid $j$ and UG $k$\\
	$\mathrm{D}_j$\> the total load of subgrid $j$\\
	$P_{\mathrm{UG},k,i}$\>the exchange power between subgrid $j$\\
	\>and UG $k$\\
	$P_{\mathrm{UG},k}$\> the output power of UG $k$ \\
	${\cal C}_j$\> agent cluster $j$\\
	$\cal F$\> the flexible agent cluster\\
	$\cal G$, $\cal V$, $\cal E$\> the communication graph, vertex set\\
	\>and edge sets\\
	${\cal G}_j$\> the communication topology of agents in cluster $j$\\
	$\cal A$, $\rho({\cal A})$\> the adjacency matrix of $\cal G$ and its spectral radius\\
	$\Omega$\>the communication weight\\
	$\mathrm{d}_i$\>the number of agent $i$'s neighbors\\
	${\tilde{\cal L}}_j$\>the Laplacian matrix corresponding to the\\
	\>communication topology of agents belong to ${\cal C}_j$\\
	$\cal B$\> the adjacency vector \\
	$\mathrm{vec}(\circ)$\> the vectorization\\
	$\mathrm{vec}^{-1}(\circ)$\> the inverse vectorization\\
	$\mathrm{den}({\tilde\theta})|_{r}$\> the densification\\
	$\mathrm{spar}({\tilde\theta})|_{r}$\> the sparsification\\
	$u_i$, ${\cal K}_1$ \> the pinning control input and its control gain\\
	$\chi_i$, ${\cal K}_2$\> the average power mismatch estimated by agent $i$\\
	$\mathrm{I}_\mathrm{n}$\>an identity matrix with $\mathrm{n}$-dimension\\
	$\overline{{\cal B}}=\mathrm{diag}\{\mathrm{\overline{b}}_{11},\cdots,\mathrm{\overline{b}}_{\mathrm{nn}}\}$\\
	\> $\mathrm{\overline{b}}_{ii}=1$ if agent $i$ is flexible; otherwise, $\mathrm{\overline{b}}_{ii}=0$\\
	${\cal Q}$\> a characteristic matrix\\
	$\Delta P_j$\> the total power deficit of subgrid $j$\\
	$P_{\mathrm{UG},k}$\> the output power of UG $k$\\
	$u_k^\mathrm{UG}$\> a pinning control input of $P_{\mathrm{UG},k}$ \\
	$\mathrm{G}=[\mathrm{g}_{k_1k_2}]$\> the adjacency matrix corresponding to the\\
	\>communication topology of UG agents\\
	$R=[r_{k_1k_2}]$\> the communication weight of $\mathrm{G}$\\
	${\cal H}=\mathrm{diag}\{\mathrm{h}_{11},\cdots,\mathrm{h}_{\mu\mu}\}$\>\\
	\> the adjacency matrix corresponding to the\\
	\>communication topology between UG agents\\
	\>and VPP agent
\end{tabbing}
\section{Introduction}
\subsection{Background and motivation}
\par \IEEEPARstart{U}{sually}, a large-scale microgrid (MG) is composed of multiple subgrids interconnected across regions, each of which contains distributed generations (DGs) and battery energy storage systems (BESSs), and is often aggregated through a virtual power plant (VPP) \cite{ESFAHANI2025115929, RUAN2024100170}. As shown in Fig. \ref{vpp}, large components of power generation equipment such as hydropower stations, wind farms \cite{AAA2024}, and thermal power generation units act as utility grids (UGs) to jointly provide frequency and voltage support for all subgrids \cite{7095586,JU2016184}. However, the distributed cluster economic dispatch (ED) problem with a VPP remains an unresolved issue, which prompts us to invest time and energy in investigating this area. In addition with regards to the power allocation scheme of DGs/BESSs and power deficit estimation scheme of each subgrid included in the standard ED problem, the cluster ED problem of a large-scale MG also includes the power allocation scheme of UGs. Moreover, the power allocation scheme in cluster economic dispatch needs to address the unique requirements of each subgrid, such as the exchange power ratio designated by the VPP between it and all UGs.
\begin{figure}[h]
	\centering
	\includegraphics[width=8cm]{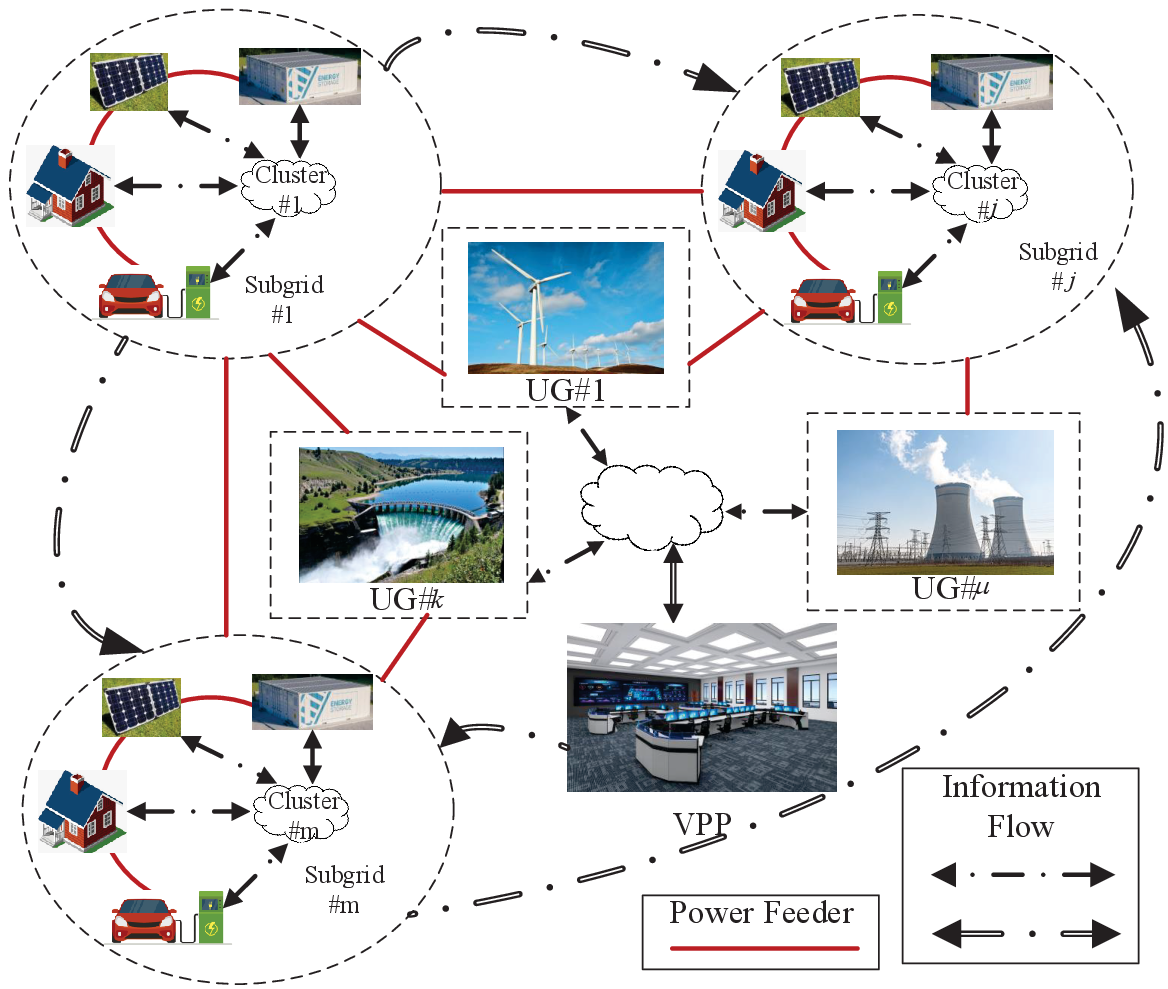}
	\caption{MGs aggregation strategy with a VPP}
	\label{vpp}
\end{figure} 
\par A large number of DGs, loads, and BESSs are often aggregated in a large-scale MG, which enhances capacity but also brings new challenges to ED. On the one hand, for a grid connected MG, all current distributed solutions are aimed at solving the ED problem of the MG as a single aggregate entity. Even for a large-scale MG containing multiple interconnected region subgrids, existing solutions can only solve the ED problem of this MG as a lowest-order aggregate \cite{JAS-2024-0296,10679615,10302354,10488096,JI20231,CHANG2024128391,JI2024127178,9928208,10693360,HBLYG2023,10379560}. However, the differences between subgrids have not been seriously considered, resulting in the unique needs of regional subgrids not being fully satisfied, i.e., most notably the differences in power demand ratios for different types of UGs by region subgrids. Determining how to achieve DG aggregation control by designing ED schemes to meet the unique needs of subgrids in different regions is an urgent problem that needs to be solved. On the other hand, the distributed ED scheme based on multi-agent systems (MASs) \cite{Moshayedi2025} relies on the communication topology \cite{9197633,9248603,8760423,9133204,liuDistributedOptimizationAlgorithm2021,9408234,9063434,9127520,8907488,8978751,zaeryDistributedGlobalEconomical2020,babazadeh-dizajiDistributedHierarchicalControl2020a,zaeryNovelFullyDistributed2021a,10102332}. Under the same controller, different types of communication topologies result in different cluster steady states, namely consensus and multi-consensus states \cite{cacaceTopologyinducedContainmentGeneral2021,monacoMulticonsensusAlmostEquitable2019,liuGroupConsensusLinear2022,gambuzzaDistributedControlMulticonsensus2021,luoClusterConsensusControl2022,tomaselliControlMulticonsensusMultiagent2024}. Especially for a large-scale MG with multiple UGs and interconnected multiple region subgrids, the communication topology determines the differences in ED results among subgrids. Exploring how communication topology affects dispatch results is a very meaningful task.
\subsection{Related work on distributed ED schemes of MGs}
\par In recent years, ED schemes based on consensus theory have been widely studied for a MG as a lowest-order aggregation. Without exception, these schemes include marginal cost (MC) \cite{10679615} consensus controllers and a routing controller to reduce electricity costs of a MG and to ensure supply-demand balance. Among them, some excellent distributed ED solutions aim to improve convergence rate \cite{10302354,10488096,JI20231,CHANG2024128391,JI2024127178} and control accuracy \cite{10302354,10488096}, save communication resources \cite{CHANG2024128391,9928208,10693360,HBLYG2023}, and protect data privacy \cite{10379560}, etc. It is not difficult to find that all of the above studies are focused on solving the ED problem of a MG as a lowest-order aggregation, which does not apply to solving the cluster ED problem of a large-scale MG. The capacity of various regions can be effectively utilized by interconnecting cross regional MGs to enhance the reliability of the power grid \cite{9197633}. However, a large-scale MG is composed of multiple interconnected regional subgrids. A DG/DESS in different subgrids belongs to different shareholders, and their power demand ratios for the same UG are different. This difference needs to be taken into account for the designed cluster distributed ED scheme.  
\par At present, cluster control schemes for MGs based on MASs are mainly divided into double-layer modes \cite{9248603,8760423,9133204,liuDistributedOptimizationAlgorithm2021} and single-layer modes \cite{9408234,9063434,9127520,8907488,8978751,zaeryDistributedGlobalEconomical2020,babazadeh-dizajiDistributedHierarchicalControl2020a,zaeryNovelFullyDistributed2021a,10102332}. 
\subsubsection{Double-layer modes for MGs}
\par A double-layer mode is actually a networked MGs (NMGs) peer-to-peer mode for cluster control of MGs \cite{9133204}, where the NMG layer is responsible for calculating the control signals to pin the MG layer. Each DG, at the MG layer, is controlled by an agent, while each MG, at the networked MG (NMG) layer \cite{9248603,8760423}, is controlled by an agent which pins a DG agent. Some researchers \cite{liuDistributedOptimizationAlgorithm2021} adopt a double-layer network to solve the cluster ED problem, thus achieving MC consensus over a connected and unweighted communication network and ensuring power supply and demand balance. Although the NMG layer requires less computing power, this method is too cumbersome for DG aggregation units because continuous communication is required between the two layers. Especially when there are a large number of aggregation units, this mode greatly increases the unnecessary complexity and computationally intensive of the ED scheme. This mode requires a high bandwidth communication link between the NMG layer and each MG layer.
\subsubsection{Single-layer modes for MGs}
\par A single-layer mode is called a master-slave scheme \cite{9408234,9127520}. As the name suggests, each group of DGs is governed by a MAS, which is composed of some master agents and slave agents. Master agents of each cluster are responsible for external communication \cite{9063434,9127520,8907488,8978751}. Researchers in \cite{9063434,9127520,8907488,8978751} advocate using a pinning control scheme, where each cluster is controlled by its master agents. In view of this, the authors elaborate on this mode in detail in \cite{10102332} and apply it to the secondary control of BESSs in the AC/DC hybrid MG. In \cite{zaeryDistributedGlobalEconomical2020}, the ED problem of cluster MGs is solved based on consensus, where an agent is selected from each cluster to be responsible for external communication. In contrast, some authors \cite{babazadeh-dizajiDistributedHierarchicalControl2020a} design a distributed cluster ED scheme to achieve consensus on MC between subgrids by controlling interconnection converters, where the communication topology within subgrids is unweighted and strongly connected. This approach is promoted as a fixed time version in \cite{zaeryNovelFullyDistributed2021a}.
\par Without exception, these researchers do not carefully consider the differences between subgrids when designing distributed ED schemes, resulting in them overlooking the special needs of each subgrid. Therefore, in this article, a general scenario is considered, where the power deficit of each subgrid is supplemented by multiple UGs, and the proportion is determined by a VPP based on its characteristics.
\subsection{State of contribution}
\par Different types of UGs and subgrids are aggregated by a VPP, which can establish exchange power ratios between each subgrid and all UGs based on the special needs of each subgrid and a subgrid at different time periods. The power mismatch, namely power deficit, of each subgrid is evenly allocated among UGs according to these ratios, which leads to MCs of DGs in different subgrids converging to different electricity prices. In other words, the solution to the ED problem of a large-scale microgrid is manifested as the MCs multiconsensus control problem and proportional allocation problem of power mismatch. In view of this, a more advanced distributed controller needs to be designed to address this cluster ED problem. Therefore, a distributed cluster ED scheme with a connected and directed communication topology are designed in this article to achieve MC multiconsensus control so that the cluster ED problem mentioned above can be solved. The specific research contents are as follows:
\begin{enumerate}	
	\item The cluster ED problem of a large-scale MG containing multiple UGs and region MGs based on multiconsensus is investigated in this article, which has been rarely reported in the past. The solution to this problem is that MCs of all DGs in each subgrid must reach a consensus and equal to the convex combination of all UG electricity prices, which is derived from the exchange power ratios defined by the VPP between a subgrid and all UGs.
	\item For a large-scale MG with a VPP aggregation strategy, where the VPP only provides the expected electricity price information for a cluster, a distributed cluster ED scheme is designed using a directed connected communication topology with well-designed communication weights, which leads to the desired MC multiconsensus. It is worth mentioning that the communication coupling between clusters is not zero.
	\item In order to ensure power supply-demand balance in each subgrid, a power mismatch estimation scheme is constructed. Subsequently, an allocation scheme is proposed by a distributed multiconsensus controller so that the power mismatch is evenly distributed among UG according to a predetermined ratio.	
\end{enumerate} 
\par The cluster ED problem based on a VPP aggregation strategy is introduced in Section II. Some of the preparatory knowledge to be used in this article is introduced in Section III. A distributed ED scheme induced by communication weights is designed in Section IV to allocate power among all UGs and DGs to ensure optimal power output and balance power supply and demand. 
In order to verify the effectiveness of the proposed solutions, some simulation cases are designed in Section V. Finally, a conclusion is drawn in Section VI.
\par \textit{Notations}: $\mathrm{n}$ denotes a positive integer, $\mathrm{I}_\mathrm{n}$ is a $\mathrm{n}$-dimension identity matrix, $\otimes$ denotes Kronecker Product, $\circ$ denotes Hadamard product. For a vector $\alpha$, $\alpha^{-1}$ means that all elements take the reciprocal. $\Vert\cdot\Vert$ represents taking the 2-norm. $\mathrm{col}(\circ)$ denotes taking a column vector.
\section{The ED Analysis of MGs under the VPP Aggregation Strategy}
\par In this article, a VPP based aggregation strategy is applied to aggregate MGs. Therefore, in this section, the optimal ED conditions for each subgrid are analyzed using the Lagrange multiplier method. Furthermore, the problem that needs to be solved is mathematically given.
\begin{figure}[h]
	\centering
	\includegraphics[width=8cm]{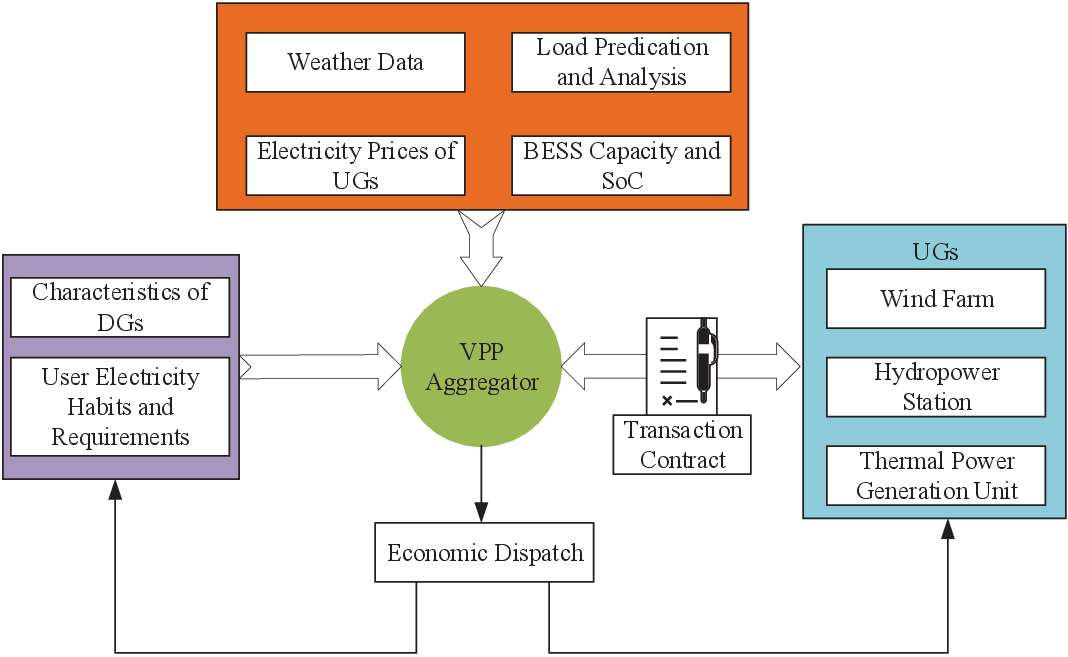}
	\caption{ED of MGs based on a VPP.}
	\label{vpp1}
\end{figure}
\par In each region, a VPP aggregates a large number of renewable energy and energy storage devices, as well as flexible loads such as controllable appliances and electric vehicles. As shown in Fig. \ref{vpp1}, as an aggregator, a VPP needs to have access to weather data, load forecast values, electricity prices of UGs, as well as the capacity and SoC of each energy storage system, in order to arrange power of each DG reasonably. By understanding the characteristics of DGs and the electricity consumption each user, a VPP can schedule controllable loads to achieve demand response. In addition, this VPP signs contracts with various power companies to purchase or sell power to ensure power supply and demand balance within the aggregator. Due to the aggregation effect of a VPP, resources in the region are efficiently utilized, making the ED scheme more flexible.
\par Considering a large-scale MG containing $\mathrm{m}$ subgrids and $\mu$ UGs aggregated by a VPP, the subordinate relationship between DGs and subgrids is as follows
$$\{\underbrace{1,\cdots,\mathrm{n}_1}_{\mathrm{subgrid}\;1},\cdots,\underbrace{\mathrm{n}_{j-1}+1,\cdots,\mathrm{n}_j}_{\mathrm{subgrid}\;j},\cdots,\underbrace{\mathrm{n}_\mathrm{m-1}+1,\cdots,\mathrm{n}_\mathrm{m}}_{\mathrm{subgrid}\;\mathrm{m}}\},$$
where $\mathrm{n}_0=0$, $j\in\{1,2,\cdots,\mathrm{m}\}$, $\mathrm{n}_j-\mathrm{n}_{j-1}$ is the number of DGs in subgrid $j$. In this article, the time variable $t$ is omitted unless necessary. Thus, a quadratic function as follow is adopted in this article as the cost function of DG $i$ in subgrid $j$, $\forall i\in\{\mathrm{n}_{j-1}+1,\cdots,\mathrm{n}_j\}$, $\forall j\in\{1,\cdots,\mathrm{m}\}$,
\begin{equation}
	C_i(P_i)=\alpha_i P_i^2+\beta_iP_i,
\end{equation}
where $P_i$ is the output power of DG $i$ in subgrid $j$, $\alpha_i$ and $\beta_i$ are the coefficients of the cost function $C_i(P_i)$ and all positive. The total cost of subgrid $j$ is 
\begin{equation}\label{costfi}
	\begin{aligned}
		&C^j(P_{\mathrm{n}_{j-1}+1},\cdots,P_{\mathrm{n}_j},P_{\mathrm{M},j})\\
		=&\sum_{i=\mathrm{n}_{j-1}+1}^{\mathrm{n}_j}C_i(P_i)+\sum_{k=1}^{\mu}\rho_kw_{jk}P_{\mathrm{M},j},
	\end{aligned}
\end{equation}
where $\rho_k$ is the electricity price of UG $k$, $P_{\mathrm{M},j}$ is the total exchange power between subgrid $j$ and UGs, $w_{jk}$ for $k\in\{1,2,\cdots,\mu\}$ is the ratio of allocating the deficit power of subgrid $j$ to UG $k$, that is, $w_{jk}P_{\mathrm{M},j}$ is the exchange power between subgrid $j$ and UG $k$, and $\sum_{k=1}^{\mu}w_{jk}=1$, $\rho_kw_{jk}$ is the aggregate electricity price for subgrid $j$.
For subgrid $j$, $\forall j\in\{1,2,\cdots,\mathrm{m}\}$, there must be 
\begin{equation}\label{balancej}
	\sum_{i=\mathrm{n}_{j-1}+1}^{\mathrm{n}_j}P_i-\mathrm{D}_j+P_{\mathrm{M},j}=0,
\end{equation}
where $\mathrm{D}_j$ is the total load power in subgrid $j$, which implies that the power supply-demand should be maintained in each subgrid.
\par In summary, invoking \eqref{costfi} and \eqref{balancej}, a Lagrangian function for subgrid $j$ is constructed as
\begin{equation}
	\begin{aligned}
		L_j=&C^j(P_{\mathrm{n}_{j-1}+1},\cdots,P_{\mathrm{n}_j},P_{\mathrm{M},j})\\
		&-\gamma_j(\sum_{i=\mathrm{n}_{j-1}+1}^{\mathrm{n}_j}P_i-\mathrm{D}_j+P_{\mathrm{M},j}),
	\end{aligned}
\end{equation}
which, by calculating the gradient, $\forall i\in\{\mathrm{n}_{j-1}+1,\cdots,\mathrm{n}_j\}$, leads to 
\begin{equation}
	\label{MCc}
	\left\{\begin{aligned}
		&\lambda_i=\gamma_j,\\
		&\sum_{k=1}^{\mu}w_{jk}\rho_k=\gamma_j,
	\end{aligned}\right.
\end{equation}
where $\lambda_i=2\alpha_iP_i+\beta_i$ is so-called the MC of DG $i$, $\gamma_j$ is a Lagrange multiplier. So, for the cluster ED problem of a large-scale MG with a VPP aggregator, there are two problems that need to be addressed.
\begin{problem} \label{p1}
	For a large-scale MG with a VPP aggregator, it is necessary to design a distributed scheme to solve the cluster ED problem. Specifically, for the MG side, there are the following issues that need to be addressed.	
	\\(\textit{P1.1}) According to \eqref{MCc}, the MCs of DGs in each subgrid should asymptotically reach consensus and converge to the expected aggregate electricity price, i.e., 
	\begin{equation}
		\lim\limits_{t\to+\infty}\lambda_i=\sum_{k=1}^{\mu}w_{jk}\rho_k,
	\end{equation}
	where $\forall i\in\{\mathrm{n}_{j-1}+1,\mathrm{n}_{j-1}+2,\cdots,\mathrm{n}_j\}$, $\forall j\in\{1,2,\cdots,\mathrm{m}\}$.
	\\(\textit{P1.2}) $\mathrm{D}_j-\sum_{i=\mathrm{n}_{j-1}+1}^{\mathrm{n}_j}P_i$ is referred to as the total power mismatch of subgrid $j$. Then, the total exchange power between each subgrid and UGs should asymptotically converge to the total power mismatch of subgrid, i.e.,
	\begin{equation}
		\lim\limits_{t\to+\infty}P_{\mathrm{M},j}=\mathrm{D}_j-\lim\limits_{t\to+\infty}\sum_{i=\mathrm{n}_{j-1}+1}^{\mathrm{n}_j}P_j,
	\end{equation}
	by which the power supply and demand of each subgrid can be maintained.
\end{problem}
\begin{problem}\label{p2}
	The total exchange power between each subgrid and UGs should be allocated to each UG proportionally, i.e.
	\begin{equation}
		\lim\limits_{t\to+\infty}P_{\mathrm{UG},k,j}=w_{ik}P_{\mathrm{M},j},
	\end{equation} 
	where $P_{\mathrm{UG},k,i}$ is the exchange power between subgrid $j$ and UG $k$. Thus, the output power of each UG should be achieved by
	\begin{equation}
		\lim\limits_{t\to+\infty}P_{\mathrm{UG},k}=\sum_{i=1}^{\mathrm{m}}w_{ik}P_{\mathrm{M},j},
	\end{equation}
	where $P_{\mathrm{UG},k}=\sum_{i=1}^{\mathrm{m}}P_{\mathrm{UG},k,j}$ is the output power of UG $k$.
\end{problem}
\section{Preliminaries}
\par Due to the selection of MASs to manage DGs in this article, some preliminary knowledge about graph theory and matrix theory, as well as some necessary definitions, are introduced here.
\subsection{Graph Theory}
\begin{definition}
	\label{de1}
	For each cluster, agents that can receive external information from another cluster are considered flexible, and they only disclose MCs of DGs they manage to other clusters. Without loss of generality, in this article, there is only one flexible agent in each cluster.
\end{definition}
\begin{definition}
	\label{defq}
	\cite{LVGFM2021} For a partition ${\cal C}=\{{\cal C}_1,\cdots,{\cal C}_m\}$ of agents, where ${\cal C}_j$ denotes agent cluster $j$, it is assumed that ${\cal C}_i\cap {\cal C}_j=\emptyset$. We define a binary matrix ${\cal Q}=[q_{ij}]$ with $\mathrm{n_m}\times \mathrm{m}$-dimension. If $i\in {\cal C}_j$, $q_{ij}=1$; $q_{ij}=0$ otherwise.
\end{definition}
\par The one key technology used to solve Problems \ref{p1} and \ref{p2} of a large-scale MG containing $\mathrm{n}_\mathrm{m}$ DGs is a communication mechanism. Specifically, the design of its distributed ED scheme based on a consensus algorithm of MASs often depends on a adjacency matrix ${\cal A}$ related to a communication topology. Agents governing DGs, which are partitioned into several clusters ${\cal C}=\{{\cal C}_1,\cdots,{\cal C}_m\}$, communicate with each other over a directed graph $\mathcal{G}(\mathcal{V},\mathcal{E})$. $\mathcal{V}=\{\mathrm{v}_1,\cdots,\mathrm{v}_{\mathrm{n_m}}\}$ denotes a set of vertices, where each vertex denotes an agent, which consists of $\mathrm{m}$ mutually exclusive sets, i.e., $\mathcal{V}_1$, $\cdots$, $\mathcal{V}_j$, $\cdots$, $\mathcal{V}_\mathrm{m}$. The subordinate relationship between agents and clusters (vertex subsets) is as follows
$$\{\underbrace{\mathrm{v}_1,\cdots,\mathrm{v}_{\mathrm{n}_1}}_{\mathcal{V}_1:\;{\cal C}_1},\cdots,\underbrace{\mathrm{v}_{\mathrm{n}_{j-1}+1},\cdots,\mathrm{v}_{\mathrm{n}_j}}_{\mathcal{V}_j:\;{\cal C}_j},\cdots,\underbrace{\mathrm{v}_{\mathrm{n}_{\mathrm{m}-1}+1},\cdots,\mathrm{v}_{\mathrm{n}_\mathrm{m}}}_{\mathcal{V}_\mathrm{m}:\;\cal{C}_\mathrm{m}}\}.$$
$\mathcal{E}=\{(\mathrm{v}_{i_1},\mathrm{v}_{i_2})|\mathrm{v}_{i_1},\mathrm{v}_{i_2}\in{\cal V}\}$ denotes a communication link set, where a directed edge $(\mathrm{v}_{i_1},\mathrm{v}_{i_2})$ implies that information flows from vertices $\mathrm{v}_{i_2}$ to $\mathrm{v}_{i_1}$. Correspondingly, the communication link set ${\cal E}$ also consists of some mutually exclusive subsets, i.e., ${\cal E}_{j_1j_2}$ for $j_1$, $j_2\in\{1,\cdots,\mathrm{m}\}$. If $j_1=j_2=j$, $\mathcal{E}_{jj}=\{(\mathrm{v}_{i_1},\mathrm{v}_{i_2})|\mathrm{v}_{i_1},\mathrm{v}_{i_2}\in{\cal V}_j\}$ denotes the communication link set of cluster $j$; otherwise, $\mathcal{E}_{j_1j_2}=\{(\mathrm{v}_{i_1},\mathrm{v}_{i_2})|\mathrm{v}_{i_1}\in{\cal V}_{j_1},\mathrm{v}_{i_2}\in{\cal V}_{j_2}\}$ denotes the set of communication links from cluster $j_2$ to cluster $j_1$. Define a adjacency matrix ${\tilde {\cal A}}_j=[\mathrm{a}_{i_1i_2}]$, where $\mathrm{v}_{i_1}$, $\mathrm{v}_{i_2}\in {\cal V}_j$. If $(\mathrm{v}_{i_1},\mathrm{v}_{i_2})\in{\cal E}_{jj}$, agent ${i_2}$ is viewed as a neighbor of agent $i_1$, and $\mathrm{a}_{i_1i_2}=1$, $\mathrm{a}_{i_1i_2}=0$ otherwise. Thus, this directed graph $\cal G$ is partitioned into $({\cal G}_1\times{\cal G}_2\times\cdots\times{\cal G}_\mathrm{m})$, where ${\cal G}_j$ is the communication topology of agents in cluster $j$. Define a adjacency matrix ${\cal A}=[\mathrm{a}_{j_1j_2}]$, where agents $\mathrm{n}_{j_1-1}+1$ and $\mathrm{n}_{j_2-1}+1$ are flexible.
\par Besides, $\mathrm{d}_{i_1}$ is the number of neighbors, namely the in-degree, of agent $i_1$. By means of this, a in-degree matrix is defined as $\mathrm{\cal D}_j=\mathrm{diag}(\mathrm{d}_{i_1})_{i_1\in{\cal C}_j}$. So far, a Laplacian matrix ${\tilde{\cal L}}_j$ corresponding to the communication topology of agents belong to ${\cal C}_j$ can be calculated as ${\tilde{\cal L}}_j={\cal D}_j-{\tilde {\cal A}}_j$.
\par For a virtual leader of MASs, according to Definition \ref{de1}, flexible agents are its only neighbors. If flexible agent $i_1\in {\cal C}_j$ can access it, it is denoted as $\mathrm{b}_j= 1$; $\mathrm{b}_j=0$ otherwise. Define a vector as $\mathrm{{\cal B}=\mathrm{col}(\mathrm{b}_1,\cdots,\mathrm{b}_{\mathrm{n_m}})}$.
\begin{definition}
	(\textit{The leader-follower cluster consensus}, LFCC) For MASs, all agents are divided into several clusters ${\cal C}=\{{\cal C}_1,\cdots,{\cal C}_m\}$, and all clusters are pinned by a cluster of leaders. Mathematically, LFCC is formulated as
	$\lim\limits_{t\to+\infty}|x_i-x_{0j}|=0,\,i\in {\cal C}_j,$
	where $x_i$ is the state of agent $i$, $x_{0j}$ is the consensus state of cluster $j$ and a convex combination of the leaders, ${\cal C}_j$ denotes cluster $j$.
\end{definition}
\subsection{Matrix Theory}
\par Before designing a LFCC controller for MCs, some necessary conclusions need to be explained in advance.
\begin{theorem}
	\label{THPF}
	\cite{VVG2017} If a matrix $\Theta$ is a non negative and irreducible square matrix, and $\rho(\Theta)$ is the largest eigenvalue of $\Theta$, then $\rho(\Theta)$ is a simple root of the characteristic equation of $\Theta$, and all components of the corresponding eigenvector are positive.
\end{theorem}
\begin{lemma}
	\label{lemeq}
	\cite{DHWAW1951} Let $\Theta$ and $\varphi$ be a matrix and vector with $\mathrm{n}$ rows, respectively. Only one of the following conclusions is true:
	\begin{enumerate}
		\item there must be a vector $x$ such that $\Theta x=\varphi$ and $x>0$,
		\item there must be a vector $y$ such that $\mathrm{\Theta}^\mathrm{T}y\ge 0$ and $y^\mathrm{T}\varphi\ge 0$.
	\end{enumerate}
\end{lemma}
\begin{definition}
	The eigenvalue and eigenvector corresponding to the spectral radius of matrix $\Theta$ are referred to as the leading eigenvalue and eigenvector, respectively.
\end{definition}
\begin{definition}
	\label{vec}
	\cite{LVGFM2021} Here, the operation of vectorizing a matrix $\Theta$ and its inverse operation are defined as $\theta=\mathrm{vec}(\Theta)$ and $\Theta=\mathrm{vec}^{-1}(\theta)$, respectively. Below are examples to illustrate these two specific operations. For a matrix $\Theta$ as follows, the vectorization $\theta=\mathrm{vec}(\Theta)$ is 
	\begin{equation}
		\Theta=\begin{bmatrix}\theta_{11}&\theta_{12}&\theta_{13}\\\theta_{21}&\theta_{22}&\theta_{23}\end{bmatrix}\Rightarrow\theta=\mathrm {vec}(\Theta)=\left[\begin{array}{c}\theta_{11}\\\theta_{21}\\\theta_{12}\\\theta_{22}\\\theta_{13}\\\theta_{23}\end{array}\right].
	\end{equation}
	Then, to obtain a $2\times 3$-dimension matrix $\Theta$ from a vector $\alpha$, the inverse vectorization $\Theta=\mathrm{vec}^{-1}(\theta)$ is defined as
	\begin{equation}
		\theta=\left[\begin{array}{c}\theta_{11}\\\theta_{21}\\\theta_{12}\\\theta_{22}\\\theta_{13}\\\theta_{23}\end{array}\right]\Rightarrow \Theta=\mathrm {vec}^{-1}(\theta)=\begin{bmatrix}\theta_{11}&\theta_{12}&\theta_{13}\\\theta_{21}&\theta_{22}&\theta_{23}\end{bmatrix}.
	\end{equation}
\end{definition}
\begin{definition}\label{denspar}
	(Densification and sparsification of a vector) For a vector $\theta$ with several zero components, a densification operator $\mathrm {den}(\theta)|_{r}$ involves removing all zero components, for example,
	\begin{equation}
		\theta=\left[\begin{array}{c}\theta_{11}\\0\\\theta_{12}\\0\\\theta_{13}\\0\end{array}\right]\Rightarrow {\tilde\theta}=\mathrm {den}(\theta)|_{r}=\left[\begin{array}{c}\theta_{11}\\\theta_{12}\\\theta_{13}\end{array}\right],
	\end{equation}
	where each element of the set $r=\{2,4,6\}$ represents an index of zero components in $\theta$. In contrast, a sparsification operator $\mathrm{spar}(\theta)|_{r}$ is used to expand a vector $\theta$ into a new vector containing some zero components, for example,
	\begin{equation}
		{\tilde\theta}=\left[\begin{array}{c}\theta_{11}\\\theta_{12}\\\theta_{13}\end{array}\right] \Rightarrow \theta=\mathrm {spar}({\tilde\theta})|_{r}=\left[\begin{array}{c}\theta_{11}\\0\\\theta_{12}\\0\\\theta_{13}\\0\end{array}\right].
	\end{equation}
\end{definition}
\begin{lemma}
	\label{nons}
	\cite{BIERKENS2014191} For a singular matrix $\cal M$ with a simple zero eigenvalue, let $\nu$ and $\omega$ be its right and left eigenvectors, respectively. Then, if and only if there exist two vectors $x$ and $y$ such that $(\nu^\mathrm{T}x)(y^\mathrm{T}\omega)\neq0$, ${\cal M}+xy^\mathrm{T}$ is non-singular.
\end{lemma}
\begin{lemma}
	\label{Mm}
	\cite{BIERKENS2014191} Base on Lemma \ref{nons}, the matrix ${\cal M}+xy^\mathrm{T}>0$ if ${\cal M}$ is an irreducible M-matrix with $\nu>0$, $\omega>0$, $x\le0$, and $y\le 0$.
\end{lemma}
\begin{lemma}
	\label{MH}
	\cite{doi:10.1137/20M131802X} A M-matrix ${\cal M}$ is Hurwitz if and only if $-{\cal M}>0$.
\end{lemma}
\section{Design of Distributed Cluster ED Scheme}
\par It needs to be emphasized again that the next problem to be addressed is the LFCC control of MCs according to (\textit{P1.1}) in Problem \ref{p1}. Specifically, in the scheme to be designed, information about their expected leader, namely aggregate electricity price, is calculated by the VPP and can be only accessed by a cluster.
Unlike before, which are designed based on MC consensus, this approach aims to find the communication weights required to achieve MC LFCC under prescribed aggregate weights.
\par Here, in order to express the proposal more clearly, it is necessary to define some terms with similar meanings. Firstly, a group represents the set of all DGs aggregated by a VPP. Secondly, a subgrid refers to a power grid consisting of load and a DG group. Finally, a cluster represents the agents that manage all DGs/DESSs in a subgrid.
\par For convenience, a necessary assumption about communication network is made here. Therein, communication topologies of intra- and inter- cluster are described in Assumption \ref{assu0}.
\begin{assumption}
	\label{assu0}
	In this paper, in any cluster, each agent communicates with its neighbors over a directed and connected graph. Between clusters, agents communicate with each other over a directed and connected graph.
\end{assumption}
\par Furthermore, the transmission mechanism of data packets collected from each subgrid, i.e., the MCs and average power mismatch information of each subgrid, on different communication links is explained by Definition \ref{assu3}, which is further deduced as Corollary \ref{assu1}.
\begin{definition}
	\label{assu3}
	Under Assumption \ref{assu0}, the average power mismatch information of each subgrid can not be disclosed to the agents managing DGs in other subgrids. That is to say, for the neighbors of agent $i$, the information packets they receive from agent $i$ may be different, which is explained by taking Fig. \ref{vpp3} as an example, where agents 5, 9, and 13 are flexible agents of cluster 1, 2 and 3, respectively.
\end{definition}
\begin{figure}
	\centering
	\includegraphics[width=8cm]{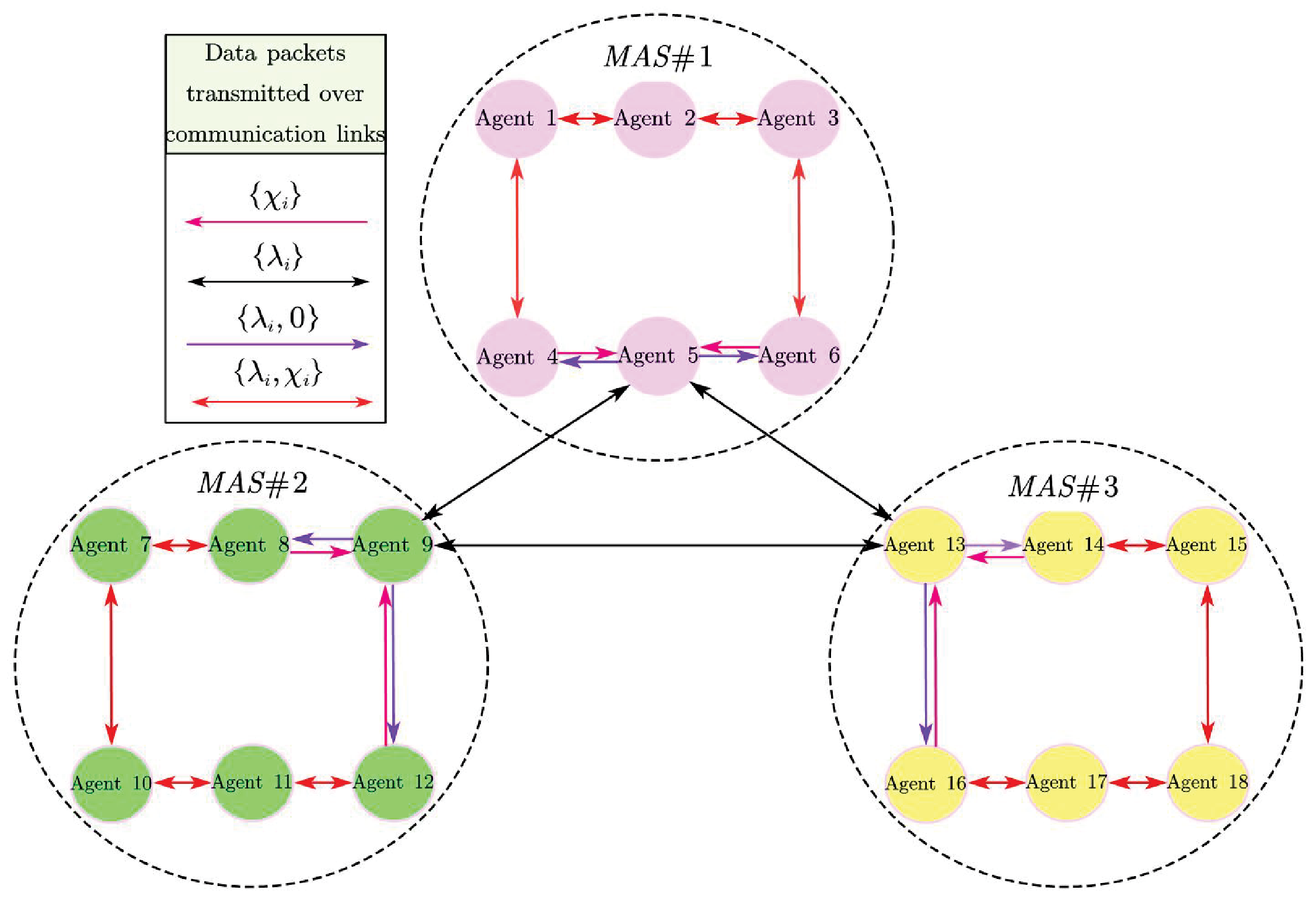}
	\caption{Data packets transmitted over different communication links.}
	\label{vpp3}
\end{figure}
\begin{corollary}
	\label{assu1}
	Based on Definitions \ref{de1} and \ref{assu3}, each flexible agent is located at the root node of the graph corresponding to the communication topology within a cluster.
\end{corollary}
\subsection{Design of Distributed LFCC for MCs}
\par According to (\textit{P1.2}) in Problem \ref{p1}, MCs of DGs in a large-scale MG should achieve multiconsensus to solve its ED problem, and their different consensus values depend on the aggregation coefficients defined by the VPP. In order to solve (\textit{P1.2}) in Problem \ref{p1}, a distributed controller over a directed and connected communication topology is designed in this section, and based on the centrality of the eigenvectors \cite{VVG2017} of the adjacency matrix, a matching communication weight matrix is designed to induce the LFCC of MCs, thus leading to the predetermined aggregation control.
\par For flexible agent $i_1$ and $i_2$, $\forall i_1,i_2\in{\cal F}$, where $\forall i_1\in{\cal F}$ and ${\cal F}\stackrel{\triangle}{=}\{1,\mathrm{n}_1+1,\cdots,\mathrm{n}_\mathrm{m-1}+1\}$, its MC control scheme is
\begin{equation}
	\label{MC}
	\dot \lambda_{i_1}={\cal K}_1(\frac{1}{\rho({\cal A})}\sum_{i_2\in {\cal F}}\mathrm{a}_{j_1j_2}\omega_{j_1j_2}\lambda_{i_2}-\lambda_{i_1})+u_{i_1},
\end{equation}
where ${\cal K}_1$ is a control gain, $u_{i_1}$ denotes the pinning control input, ${\cal A}=[\mathrm{a}_{j_1j_2}]$ denotes the adjacency matrix corresponding to the communication topology among flexible agents, and $\rho({\cal A})$ denotes the spectrum of ${\cal A}$, $\Omega=[\omega_{j_1j_2}]$ is the communication weight matrix to be deigned. The proposed pinning control input is designed as
\begin{equation}
	\label{ci1}
	u_{i_1}={\cal K}_1\mathrm{b}_{i_1}(W_{j_1}\rho-\lambda_{i_1}),
\end{equation}
where $W_{j_1}\rho$ represents the aggregate electricity pricing strategy formulated by the VPP, DG $i_1$ belongs to subgrid $j$, $W_{j_1}$ is the $j_1$-th row of $W=[w_{j_1j_2}]$, $\mathrm{B}=\mathrm{diag}(\mathrm{b}_1,\cdots,\mathrm{b}_{\mathrm{m}})$. Then, a compact form of \eqref{MC} is
\begin{equation}
	\label{lam1}
	\dot \lambda_{\cal F}={\cal K}_1\mathrm{A}\lambda_{\cal F}+u,
\end{equation}
where $\lambda_{\cal F}=\mathrm{col}(\lambda_{i_1})_{i_1\in{\cal F}}$,  $\mathrm{A}=\frac{1}{\rho({\cal A})}({\cal A}\circ \Omega)-\mathrm{I}_{\mathrm{n}}$, $u=\mathrm{col}(u_{i_1})_{i_1\in{\cal F}}$. The pinning control input vector is
\begin{equation}
	\label{ci}
	u={\cal K}_1\mathrm{B}(W\rho-\lambda_{\cal F}).
\end{equation} 
\par For other agents in cluster $j$, such as agent $i_1$, $\forall i_1\in\{\mathrm{n}_{j-1}+2,\cdots,\mathrm{n}_j\}$, its MC control scheme is
\begin{equation}
	\label{MC0}
	\dot \lambda_{i_1}={\cal K}_1\sum_{i_2\in{\cal C}_j}\mathrm{{\tilde a}}_{i_1i_2}(\lambda_{i_2}-\lambda_{i_1}),
\end{equation}
where ${\tilde {\cal A}}_j=[\mathrm{{\tilde a}}_{i_1i_2}]$ denotes the adjacency matrix corresponding to the communication topology among agents belonging to ${\cal C}_j$.
Then, a compact form of \eqref{MC0} is
\begin{equation}
	\label{lam0}
	\dot {\tilde\lambda}_j=-{\cal K}_1{\tilde {\cal L}}_j{\tilde\lambda}_j,
\end{equation}
where ${\tilde\lambda}_j=\mathrm{col}(\lambda_{i_1})_{i_1\in{\cal C}_j}$, ${\tilde {\cal L}}_j$ is the Laplacian matrix corresponding to cluster $j$. 
\begin{remark}
	From \eqref{lam1} with \eqref{ci}, it can be seen that the next important task is to design a method to obtain a matching communication weight matrix to reach the desired LFCC. As a complex cyber physical system, a reasonable communication network need to be designed to solve the cluster ED problem of a large-scale MG. In other words, because the collaborative effect of DG groups heavily relies on communication infrastructure, the characteristics of communication networks usually determine the steady-state of cluster networks, such as consensus and multiconsensus. So, starting from a graph theory perspective, we use matrix theory and the centrality of eigenvectors\cite{VVG2017} to construct an effective method for obtaining the communication weight matrix, ensuring that the expected attributes of the communication network can be assigned, thus achieving MC LFCC of MGs.
\end{remark}
\par Based on Assumption \ref{assu0} and according to Theorem \ref{THPF}, ${\cal A}$ is a non-negative and irreducible matrix. Given a eigenvector $v_1^{\prime}>0$, a lemma and its proof are presented to obtain a suitable weight matrix $\Omega$ and prove that $v_1^{\prime}$ is the leading eigenvector of ${\cal A}\circ \Omega$.
\begin{lemma}
	\label{lemwt}
	A weight matrix $\Omega$ can be found such that $v_1^{\prime}$ is the leading eigenvalue of ${\cal A}\circ \Omega$ associated with $\rho({\cal A})$.
\end{lemma}
\begin{proof}
	If the following holds,
	\begin{equation}
		({\cal A}\circ \Omega)v_1^{\prime}=\rho({\cal A})v_1^{\prime},
	\end{equation}
	$v_1^{\prime}$ is the leading eigenvalue of ${\cal A}\circ \Omega$ associated with $\rho({\cal A})$. The above can be vectorized, by Definition \ref{vec}, into
	\begin{equation}
		\label{M1}
		({v_1^{\prime}}^\mathrm{T}\otimes \mathrm{I}_{\mathrm{m}})(\alpha\circ \omega)=\rho({\cal A})v_1^{\prime},
	\end{equation}
	where $\alpha=\mathrm{vec}({\cal A})$ and $\omega=\mathrm{vec}(\Omega)$. Obviously, $\alpha$ is a binary vector, whose component can only be 0 or 1. By letting $M=({v_1^{\prime}}^\mathrm{T}\otimes \mathrm{I}_{\mathrm{m}})$, $M(\alpha\circ \omega)={\hat M}{\hat\omega}$, where ${\hat\omega}=\mathrm{den}(\alpha\circ\omega)|_r$, ${\hat M}$ is obtained by deleting the columns of $M$ that belong to the set $r$ in the index.
	Then, \eqref{M1} is densified into
	\begin{equation}
		\label{sle}
		{\hat M}{\hat \omega}=\rho({\cal A})v_1^{\prime}.
	\end{equation}
	By finding the solution to the linear equation system above, it can be expanded into a proper matrix $\Omega$ such that $v_1^{\prime}$ is the leading eigenvector associated with $\rho({\cal A})$. 
	\par Thus, according to the Perron–Frobenius theorem in Theorem \ref{THPF}, $v_1^{\prime}>0$, after which is follows that each column of $M$ has only one a non-zero component. According to the construction method of matrix $\hat M$, it can be asserted that each column of $\hat M$, which is a non-negative matrix, has only one non-zero component. Thus, the system of linear equations \eqref{sle} has the following properties:
	\begin{enumerate}
		\item each linear equation only has one unknown variable,
		\item the number of linear equations is not greater than those of unknown variables, which implies \eqref{sle} is determined or underdetermined.
	\end{enumerate}
	However, according to Lemma \ref{lemeq} and $\rho({\cal A})v_1^{\prime}>0$, the system of linear equations \eqref{sle} has a non-negative solution. 
	\par Then, by the inverse vectorization defined in Definition \ref{vec} and the sparsification operation $\omega=\mathrm{spar}(\hat\omega)|_r$ defined in the Definition \ref{denspar}, a desired weight matrix $\Omega$ can be obtained by $\Omega=\mathrm{vec}^{-1}(\omega)$. So far, Lemma \ref{lemwt} is proven. $\hfill\blacksquare$
\end{proof}
\begin{remark}
	For \eqref{sle}, it is a system of linear equations about $\hat\omega$, whose has at least one solution. And 
	\begin{enumerate}
		\item if any two components of vector $v_1^{\prime}$ are not the same, \eqref{sle} is determined,
		\item otherwise, it is underdetermined. 
	\end{enumerate}
	Considering that equation \eqref{sle} is underdetermined, a linear programming model can be constructed as follows to solve the optimal solution,
	\begin{equation}\label{opti}
		\min\,f^\mathrm{T}{\hat\omega},\,\mathrm{s.t.}\,{\hat M}{\hat \omega}=\rho({\cal A})v_1^{\prime},
	\end{equation}
	where there are multiple criteria to choose coefficient $f$. For example, $f_i=1$ for $\forall i$ means that the solution to the optimization problem in \eqref{opti} ensures that the communication topology has the minimum total communication weight.
\end{remark} 
\par Here, an optimization objective \eqref{opti1} is proposed to find the minimum norm solution to \eqref{sle}, i.e.,
\begin{equation}\label{opti1}
	\min\,\Vert{\hat\omega}\Vert,\,\mathrm{s.t.}\,{\hat M}{\hat \omega}=\rho({\cal A})v_1^{\prime}.
\end{equation}
which can be solved by
\begin{equation}
	\label{Wm}
	{\hat\omega}={\hat M}^+\rho({\cal A})v_1^{\prime},
\end{equation}
where ${\hat M}^+$ is the Moore–Penrose inverse of ${\hat M}$.
\par The following theorem and its proof are provided to prove that MCs can achieve the desired LFCC under \eqref{lam1} with \eqref{ci}.
\begin{theorem}
	\label{TH1}
	For $\mathrm{m}$ clusters MGs containing $\mathrm{n}_\mathrm{m}$ DGs, MCs can be achieved desired LFCC under \eqref{lam1} and \eqref{lam0} with \eqref{ci} and $\mathrm{k}_1>0$, where $\Omega$ is obtained by Lemma \ref{lemwt} and \eqref{Wm}.
\end{theorem}
\begin{proof}
	The following is obtained by substituting \eqref{ci} into \eqref{lam1},
	\begin{equation}
		\dot \lambda_1={\cal K}_1(\mathrm{A}-\mathrm{B})\lambda_1+{\cal K}_1\mathrm{B}W\rho.
	\end{equation}
	We define control error as $e^\lambda=\lambda_1-W\rho$. Thus, according to $\dot\lambda$ and by taking the derivative of $e^\lambda$ with respect to $t$, one can get
	\begin{equation}
		\label{lame}
		\dot e^\lambda={\cal K}_1(\mathrm{A}-\mathrm{B})e^\lambda.
	\end{equation}
	Obviously, $\lim\limits_{t\to+\infty}\|e^\lambda\|=0$ if $(\mathrm{A}-\mathrm{B})$ is Hurwitz and ${\cal K}_1>0$, which will be confirmed as follows.
	\par Define two matrices as follows,
	\begin{equation}
		\label{At}
		\mathrm{\tilde A}=\rho({\cal A})\mathrm{I}_\mathrm{m}-({\cal A}\circ \Omega)=-\rho({\cal A})\mathrm{A},
	\end{equation}
	\begin{equation}
		\mathrm{F}=\mathrm{\tilde A}+\rho({\cal A})\mathrm{bb^\mathrm{T}}=-\rho({\cal A})(\mathrm{A}-\mathrm{B}).
	\end{equation}
	As the matrix $({\cal A}\circ \Omega)$ is non-negative and irreducible, according to \eqref{At}, there exist that the right and left eigenvectors associated with zero eigenvalue are positive, which and $\mathrm{b}$ satisfy the result in Lemma \ref{nons}, i.e., $\mathrm{F}$ is non-singular. Then, according to Lemma \ref{Mm}, $\mathrm{F}$ has all positive leading minors, i.e., is positive. Thus, $({\cal A}-\mathrm{B})$ is Hurwitz according to Lemma \ref{MH}, which leads to $\lim\limits_{t\to+\infty}\|e^\lambda\|=0$. 
	\par Next, based on Corollary \ref{assu1} and the above result, MCs of DGs can be achieved LFCC driven by \eqref{lam0}.
	\par So far, Theorem \ref{TH1} is proven. $\hfill\blacksquare$
\end{proof}
\par Here, the differences from previous literature should be emphasized in the following two remarks.
\begin{remark}
	In previous works \cite{7556330,CHANGBIN20141965,7835158,4663649,YU20142341}, the Laplacian matrix corresponding to the communication topology between agents follows the following form
	\begin{equation}
		{\cal L}=\begin{bmatrix}
			c_1{\cal L}_{11}& {\cal L}_{12} & \cdots &{\cal L}_{1n}\\
			{\cal L}_{21}& c_2{\cal L}_{22} & \cdots & {\cal L}_{2n}\\
			\cdots & \cdots & \ddots & \cdots\\
			{\cal L}_{n1}& {\cal L}_{n2} & \cdots & c_n{\cal L}_{nn}\\
		\end{bmatrix}.
	\end{equation}
	Furthermore, in \cite{QIN20132898,XIA20112395,9388882}, the communication topology is required to follow a block Laplacian matrix as
	\begin{equation}
		{\cal L}=\begin{bmatrix}
			{\cal L}_{11}& {\cal L}_{12} & \cdots &{\cal L}_{1n}\\
			{\cal L}_{21}& {\cal L}_{22} & \cdots & {\cal L}_{2n}\\
			\cdots & \cdots & \ddots & \cdots\\
			{\cal L}_{n1}& {\cal L}_{n2} & \cdots & {\cal L}_{nn}\\
		\end{bmatrix}.
	\end{equation}
	The above two Laplacian matrices are clustered into blocks. From previous results, it can be seen that when the communication coupling strength between clusters needs to be 0, i.e., ${\cal L}_{ij}1_{m_j}=0$, $i\neq j$, and/or the communication coupling strength $c_i$ within cluster $i$ is sufficiently large, cluster consensus can naturally be achieved. In this article, this strict condition is broken. The communication coupling strength between clusters is not required to be 0. In fact, the property of non-zero inter cluster communication coupling is utilized to achieve the control objective in this paper.
\end{remark}
\begin{remark}
	Although some work has utilized pinning control techniques to achieve cluster consensus, our work makes significant changes. (a) Only a few agents in a few clusters can access the leader, and the communication topology between clusters is connected. The designed scheme follows the form of a leader-following consensus controller of MASs, where only a few agents can access the leader information, and not all clusters can access the leader information. A special case is that only one agent in a certain cluster can access the leader information. (b) The communication coupling between clusters is used to achieve cluster consensus, and the total coupling strength here does not necessarily need to be 0, which is the biggest difference from previous cyclic/acyclic communication topologies \cite{QIN20132898,XIA20112395}.
\end{remark}
\par In this section, the power allocation problem of DGs has been solved. Subsequently, by planning the communication weights to design distributed controllers, the power deficit of each subgrid will be supplemented by UGs according to a predetermined ratio.
\begin{figure*}
	\centering
	\includegraphics[width=16cm]{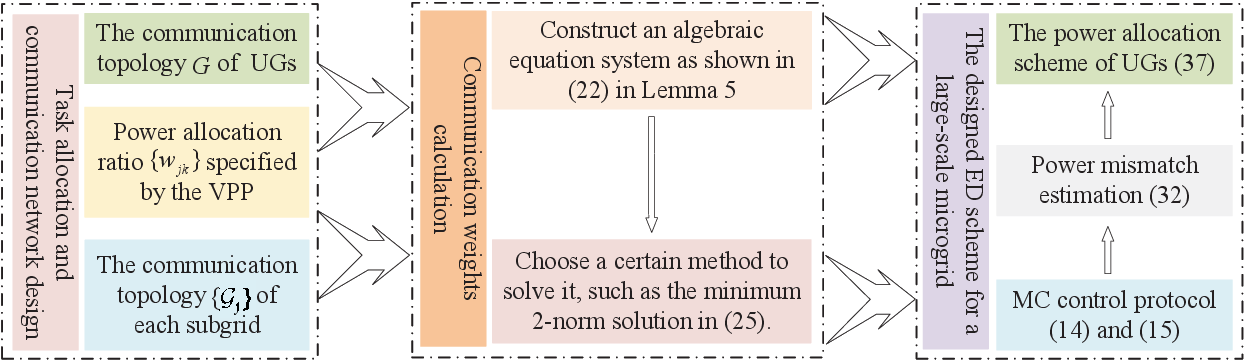}
	\caption{A schematic diagram of the designed scheme.}
	\label{fc}
\end{figure*}
\subsection{Estimation of Power Mismatch of Each Subgrid}
\par To estimated the total power deficit, a distributed estimation scheme is developed as follows
\begin{equation}
	\label{pm1}
	\dot\chi_i={\cal K}_2\sum_{i,l\in{\cal C}_j}\mathrm{{\tilde a}}_{il}(\chi_l-\chi_i)-\frac{1}{2\alpha_i}(1-\mathrm{\overline{b}}_{ii})\dot\lambda_i,
\end{equation}
where $\chi_i$ and $\chi_l$ are the estimated average power mismatches of subgrid $j$, $\mathrm{\overline{b}}_{ii}=1$ if agent $i$ is flexible; otherwise, $\mathrm{\overline{b}}_{ii}=0$, and ${\cal K}_2>0$ is a control gain, whose compact form is
\begin{equation}
	\label{pm2}
	\dot\chi=-{\cal K}_2{\tilde {\cal L}}\chi-\frac{1}{2}\alpha^{-1}\circ((\mathrm{I_n}-\overline{{\cal B}})\dot\lambda),
\end{equation}
where $\chi=\mathrm{col}\{\chi_1,\cdots,\chi_n\}$, $\overline{{\cal B}}=\mathrm{diag}\{\mathrm{\overline{b}}_{11},\cdots,\mathrm{\overline{b}}_{\mathrm{nn}}\}$.
With these prerequisites, the vector composed of the total mismatch of each subgrid conforms to the following dynamics
\begin{equation}
	\label{ER1}
	\dot P_\mathrm{M}=-{\cal K}_2{\cal Q}^\mathrm{T}{\tilde{\cal B}}\chi-{\cal Q}^\mathrm{T}\frac{1}{2\alpha}\circ (\overline{{\cal B}}\dot\lambda),
\end{equation}
where ${\cal Q}$ is a characteristic matrix following Definition \ref{defq}, and $P_\mathrm{M}$ is an $\mathrm{m}$-dimensional vector. It is necessary for the initial condition to satisfy ${\cal Q}^\mathrm{T}(\chi(0)-\Delta P(0))+P_\mathrm{M}(0)=0_\mathrm{m}$. 
However, in practical applications, it is necessary to ensure supply-demand balance before ED. Since the power deficit of each subgrid is borne by UGs, the initial conditions to choose $P_M(0)-{\cal Q}^\mathrm{T}\Delta P(0)=0_\mathrm{m}$ and $\chi(0)=0_\mathrm{n}$ appear more practical.
\par Thus, there is a theorem as follows on $\chi(0)$ and $P_\mathrm{M}(0)$ under \eqref{pm2} and \eqref{ER1}.
\begin{theorem}
	\label{TH21}
	Driven by \eqref{pm2} and \eqref{ER1}, under the initial conditions $P_\mathrm{M}(0)-{\cal Q}^\mathrm{T}\Delta P(0)=0_\mathrm{m}$ and $\chi(0)=0_\mathrm{n}$, the estimated average power mismatch of each subgrid asymptotically converges to 0 if $\mathrm{k}_2>0$. In addition, the total power mismatch within each subgrid can be estimated.
\end{theorem}
\begin{proof}
	According to Corollary \ref{assu1} and Definition \ref{de1}, in any cluster $j$, the flexible agent is located at the root node of the graph ${\cal G}_j$. Therefore, in view of $\lim\limits_{t\to+\infty}\|\dot\lambda\|=0$ according to Theorem \ref{TH1}, it can be asserted that
	\begin{equation}
		\lim\limits_{t\to+\infty}\|\chi\|=0.
	\end{equation}
	\par Based on \eqref{pm2} and \eqref{ER1}, there must be
	\begin{equation}
		\begin{aligned}
			{\cal Q}^\mathrm{T}(\dot\chi-\Delta \dot P)+\dot P_\mathrm{M}=&{\cal Q}^\mathrm{T}(\dot\chi+\frac{1}{2\alpha}\circ\dot\lambda)+\dot P_\mathrm{M}\\
			=&-{\cal K}_2{\cal Q}^\mathrm{T}\mathrm{L}\chi=0_\mathrm{m},
		\end{aligned}
	\end{equation}
	where $\mathrm{L}={\tilde {\cal L}}+{\tilde {\cal B}}$. With the conditions that $P_\mathrm{M}(0)-{\cal Q}^\mathrm{T}\Delta P(0)=0_\mathrm{m}$, $\chi(0)=0_\mathrm{n}$ and $\lim\limits_{t\to+\infty}\|\chi\|=0$, one can get $\lim\limits_{t\to+\infty}(P_\mathrm{M}-{\cal Q}^\mathrm{T}\Delta P)=0_\mathrm{m}$, which implies that $\lim\limits_{t\to+\infty} P_\mathrm{M}$ is the power deficit vector. 
	\par Besides, the result of Theorem \ref{TH1} shows that $\Delta P_i$ for $i\in\{1,2,\cdots,\mathrm{n}\}$ is asymptotically stable. Thus, $P_{\mathrm{M},j}$ for $j\in\{1,2,\cdots,\mathrm{m}\}$ is asymptotically stable. $\hfill\blacksquare$
\end{proof}
\par Thus, based on the above results, a proposition is proposed here to summarize the current conclusion.
\begin{proposition}
	According to the results of Theorems \ref{TH1} and \ref{TH21}, Problem \ref{p1} can be solved.
\end{proposition}
\subsection{Design of A Power Allocation Scheme Between UGs}
\par Next, a power allocation scheme is designed as follows for UGs to compensate for the power deficit of each subgrids,
\begin{equation}
	\label{agg}
	\dot P_{\mathrm{UG},k_1}={\cal K}_2(\frac{1}{\rho( \mathrm{G})}\sum_{k_2=1}^{\mu}\mathrm{g}_{k_1k_2}r_{k_1k_2} P_{\mathrm{UG},k_2}-P_{\mathrm{UG},k_1})+u_{k_1}^{\mathrm{UG}},
\end{equation}
where $\mathrm{G}=[\mathrm{g}_{k_1k_2}]$ and $R=[r_{k_1k_2}]$ are an adjacency and a weight matrices corresponding to the communication topology between UG agents, $\rho(\mathrm{G})$ is the leading eigenvalue of $\mathrm{G}$, $k_1$, $k_2\in\{1,2,\cdots,\mu\}$, $u_{k_1}^{\mathrm{UG}}$ is the pinning control input, and
\begin{equation}
	\label{ci2}
	u_{k_1}^{\mathrm{UG}}={\cal K}_2\mathrm{h}_{{k_1k_1}}(W_{k_1}^\mathrm{T}P_\mathrm{M}-P_{\mathrm{UG},{k_1}}),
\end{equation}
${\cal H}=\mathrm{diag}\{\mathrm{h}_{11},\cdots,\mathrm{h}_{\mu\mu}\}$ is an adjacency matrix, $W_{k_1}^\mathrm{T}$ is the ${k_1}$-th row of $W^\mathrm{T}$. $W_{k_1}^\mathrm{T}P_\mathrm{M}$ is the reference signal of UG agent ${k_1}$ provided by the VPP agent. Therefore, there is the following theorem regarding the allocation of output power for each UG.
\begin{theorem}
	\label{TH2}
	There can be definitely found a weight matrix $R$ to ensure that, driven by \eqref{agg} and \eqref{ci2}, the output power $P_{\mathrm{UG},k_1})$ of UG $k_1$ for $\forall k_1\in\{1,2,\cdots,\mu\}$ can be reached LFCC by \eqref{agg} amd \eqref{ci2}.
\end{theorem}
\begin{proof}
	Similar to Theorem \ref{TH1}, by Lemma \ref{lemwt}, a weight matrix $R$ can be found such that Theorem \ref{TH2} holds. $\hfill\blacksquare$
\end{proof}
\par Thus, based on the above results, a proposition is proposed here to summarize the current conclusion.
\begin{proposition}
	According to the results of Theorem \ref{TH2}, Problem \ref{p2} can be solved.
\end{proposition}
\subsection{Implementation of the Designed Scheme and Analysis of Its Computational Complexity}
\par A concise diagram that facilitates the implementation of the designed scheme is shown in Fig. \ref{fc}. Based on the given communication topologies and power allocation ratios specified by a VPP, using Lemma \ref{lemwt} and \eqref{Wm}, appropriate communication weights can be calculated such that the weighted adjacency matrix has a given eigenvalue centrality. Then, the MC multiconsensus controller, power mismatch estimator, and power allocation protocol are designed by using this communication network and its weighted version. So far, the ED problem of a large-scale microgrid has been solved.
\par In Lemma \ref{lemwt}, the computational complexity of constructing a linear system of equations as shown in \eqref{sle} and solving it using \eqref{Wm} are ${\cal O}(\mathrm{n_m}^2)$ and ${\cal O}(\mathrm{m}^3)$, respectively. However, since the designed dynamic based ED scheme are all linear, computational complexities of \eqref{MC} with \eqref{ci1}, \eqref{pm2}, \eqref{ER1}, and \eqref{agg} with \eqref{ci2} are ${\cal O}(\mathrm{n_m}^3)$, ${\cal O}(\mathrm{n_m}^3)$, ${\cal O}(\mathrm{n_m}^3)$, and ${\cal O}(\mathrm{m}^3)$, respectively. Thus, the computational complexity of the entire scheme is ${\cal O}(\mathrm{n_m}^3)$.
\begin{figure*}[htbp]
	\centering
	\begin{subequations}
		\label{matr}
		\begin{equation}
			{\tilde {\cal L}}=\left[\begin{matrix}
				0&0&0&0\\
				-1&3&-1&-1\\
				-1&-1&3&-1\\
				-1&-1&-1&3\\
			\end{matrix}\right],{\tilde {\cal B}}=\left[\begin{matrix}
				3&-1&-1&-1\\
				0&0&0&0\\
				0&0&0&0\\
				0&0&0&0\\
			\end{matrix}\right],{\cal B}=\left[\begin{matrix}
				1&0&0&0\\
				0&0&0&0\\
				0&0&0&0\\
				0&0&0&0
			\end{matrix}\right],
			G=\left[\begin{matrix}
				0&1&0\\
				0&0&1\\
				1&0&0
			\end{matrix}\right].
		\end{equation}
		\begin{equation}
			\label{Wt1}
			\Omega=\left[\begin{matrix}
				0&0&0.4651&0&0&0&0&0\\
				0.1395&0&0&0.4186&0&0&0&0\\
				0&0&0&0.9660&0&0&0.6440&0\\
				0&0&0&0&0&0&0&2.0930\\
				5.5813&0&0&0&0&0&0&0\\
				0&0.3283&0&0&1.3133&0&0&0\\
				0&0&0&0&0.5581&0&0&0.2791\\
				0&0&0&0&0&0.6977&0&0
			\end{matrix}\right],R=\left[\begin{matrix}
				0&0.6&0\\
				0&0&2.5\\
				0.6667&0&0
			\end{matrix}\right].
		\end{equation}
		\begin{equation}
			{\cal A}=\left[\begin{matrix}
				0&0&1&0&0&0&0&0\\
				1&0&0&1&0&0&0&0\\
				0&0&0&1&0&0&1&0\\
				0&0&0&0&0&0&0&1\\
				1&0&0&0&0&0&0&0\\
				0&1&0&0&1&0&0&0\\
				0&0&0&0&1&0&0&1\\
				0&0&0&0&0&1&0&0
			\end{matrix}\right],B=\left[\begin{matrix}
				1&0&0&0&0&0&0&0\\
				0&0&0&0&0&0&0&0\\
				0&0&0&0&0&0&0&0\\
				0&0&0&0&0&0&0&0\\
				0&0&0&0&0&0&0&0\\
				0&0&0&0&0&0&0&0\\
				0&0&0&0&0&0&0&0\\
				0&0&0&0&0&0&0&0
			\end{matrix}\right].
		\end{equation}
	\end{subequations}
\end{figure*}
\begin{figure}[h]
	\centering
	\includegraphics[width=8cm]{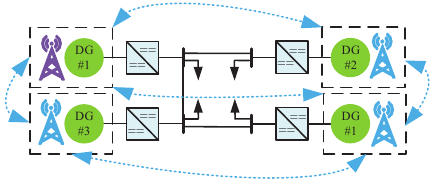}
	\caption{A 4 bus system for each subgrid.}
	\label{4bus}
\end{figure}
\section{Case Study}
A MG is selected as a subgrid as shown in Fig. \ref{4bus}, where agent 1 marked in light purple managing DG 1 is chosen as the flexible agent. The transmission mechanism of data packets follows Definition \ref{assu3}. For the sake of simplicity and generality, a VPP is leveraged to aggregate 8 such power network, where all flexible agents communicate with each other over a graph shown in Fig. \ref{com1} and the pinning signal is provided by a VPP. So, according to Figs. \ref{4bus} and \ref{com1}, the matrices used are shown in \eqref{matr}.
\begin{figure}[h]
	\centering
	\includegraphics[width=4cm]{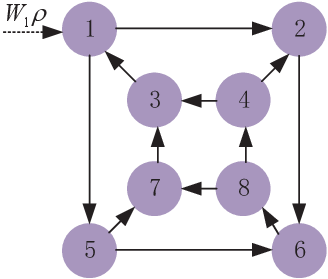}
	\caption{Communication topology among flexible agents.}
	\label{com1}
\end{figure}
\par To verify the effectiveness of the designed scheme and further test it, four simulation cases are designed under different detailed configurations. In Case 1, the designed scheme is tested under constant electricity prices, which indicates that under the established power allocation ratio, MC can achieve the expected multiconsensus, and the power deficit of each subgrid can also be supplemented according to these ratios. Furthermore, under a tiered electricity price, this ED problem can still be solved in Case 2. Even in the event of load switching, as shown in Case 3, the change in load power can be allocated to UGs according to the predetermined ratio. In addition, in Case 4, the designed scheme is also tested when the aggregation coefficient changes, which shows that when the aggregation coefficient changes, the power deficit in each subgrid is allocated according to the new proportion, and MC also achieved the desired multiconsensus.
\begin{figure}[h]
	\centering
	\includegraphics[width=4cm]{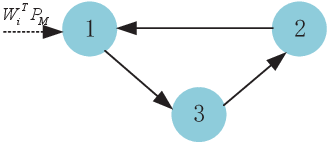}
	\caption{Communication topology among UG agents.}
	\label{com2}
\end{figure}
\begin{figure}[h]
	\centering
	\includegraphics[width=8cm]{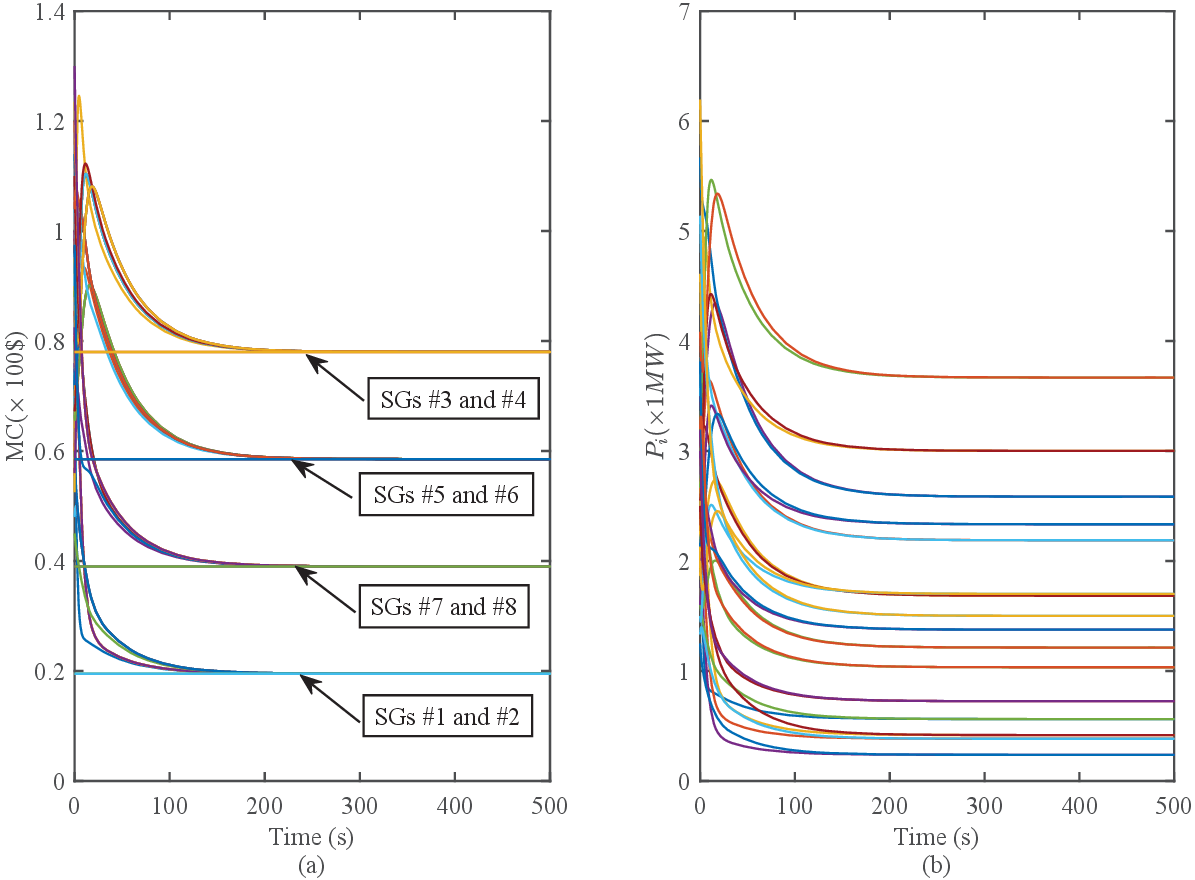}
	\caption{Evolution of MC and output power of each DG with a constant electricity price.}
	\label{MC_P}
\end{figure}
\begin{figure}[h]
	\centering
	\includegraphics[width=8cm]{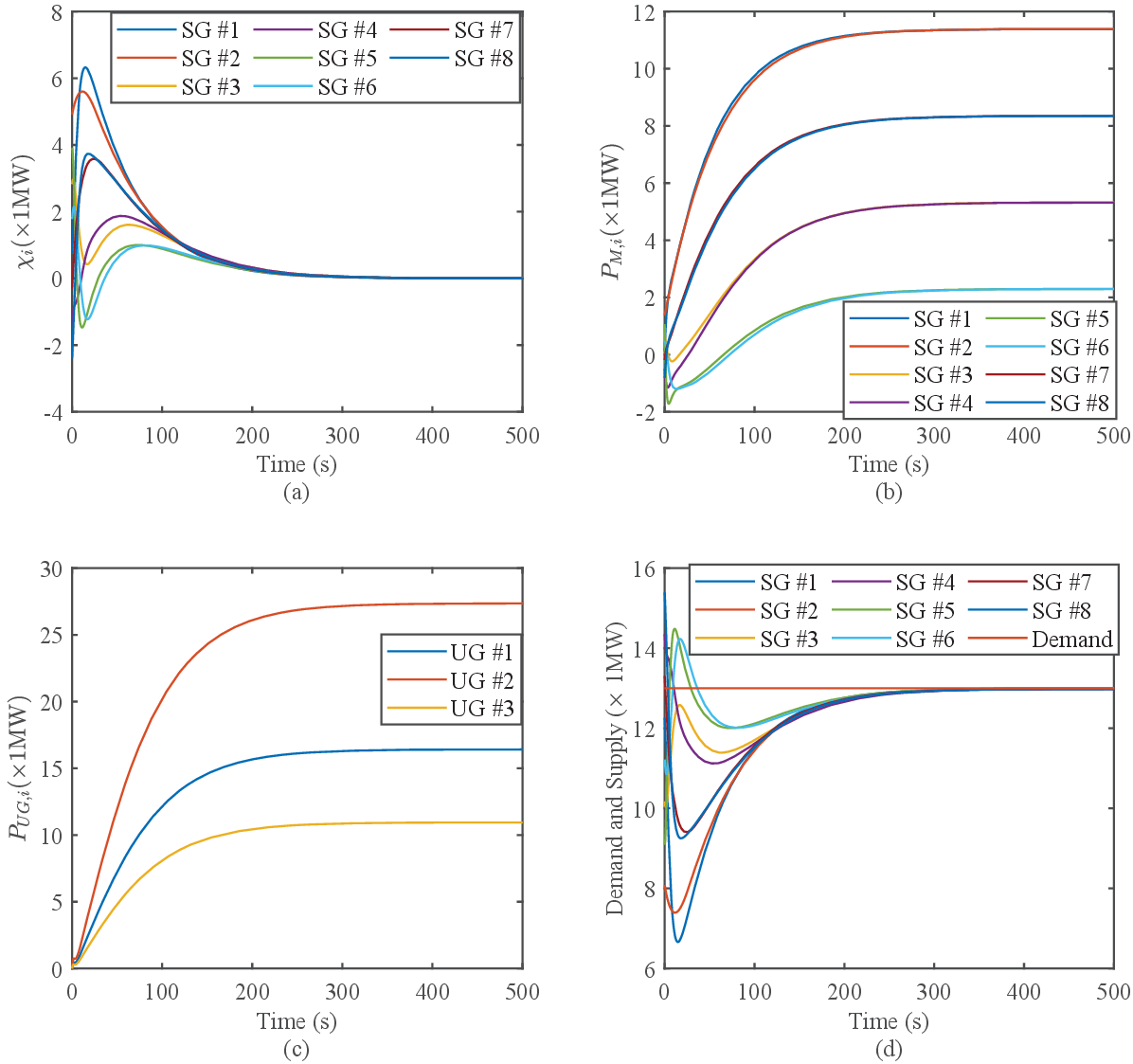}
	\caption{Evolution of total estimated power mismatch, exchange power, and supply-demand balance of each subgrid and the output power of each UG under constant electricity price.}
	\label{UG}
\end{figure}
\subsection{Case 1. Simulation results under constant electricity prices}
\par The weight matrix used in the aggregation strategy is
\begin{equation}
	\label{W1}
	W=\left[\begin{matrix}
		0.71&0.27&0.02\\
		0.79&0.13&0.08\\
		0.05&0.45&0.5\\
		0.09&0.38&0.53\\
		0.016&0.022&0.962\\
		0&0.05&0.95\\
		0.18&0.71&0.11\\
		0.3&0.5&0.2
	\end{matrix}\right].
\end{equation}
The adopted electricity price vector is $\rho=[0.1,0.4,0.8]^\mathrm{T}$. Then, MCs of each subgrid are expected to converge to $W\rho=[0.195,0.195,0.585,0.585,0.78,0.78,0.39,0.39]^\mathrm{T}$. With this data in hand, using Lemma \ref{lemwt} to obtain the communication weight matrices as shown in \eqref{Wt1} and applying schemes \eqref{lam1}, \eqref{ci}, \eqref{pm2}, \eqref{ER1}, \eqref{agg}, and \eqref{ci2}, simulation results are shown in Figs. \ref{MC_P} and \ref{UG}.
\begin{figure}[h]
	\centering
	\includegraphics[width=8cm]{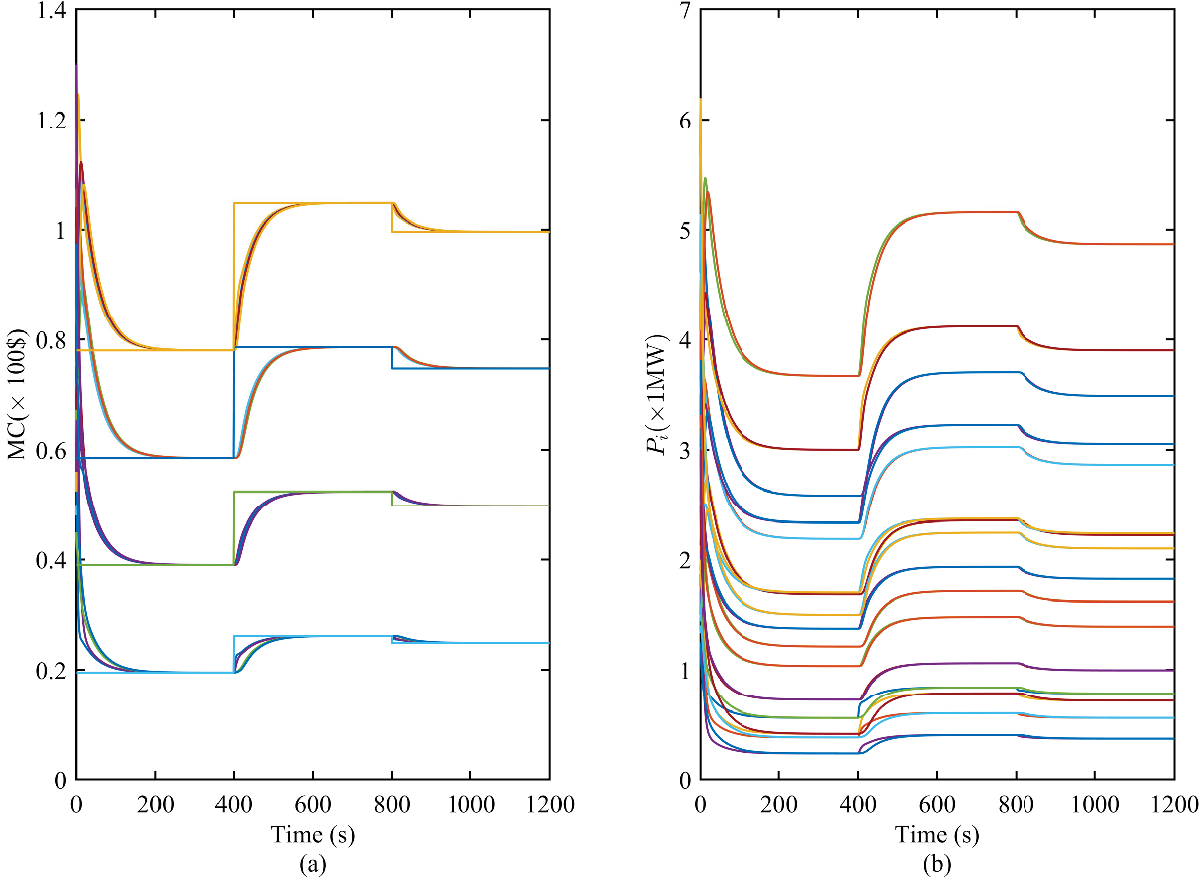}
	\caption{Evolution of MC and output power of each DG under tiered electricity price.}
	\label{MC_PJieti}
\end{figure}
\begin{figure}[h]
	\centering
	\includegraphics[width=8cm]{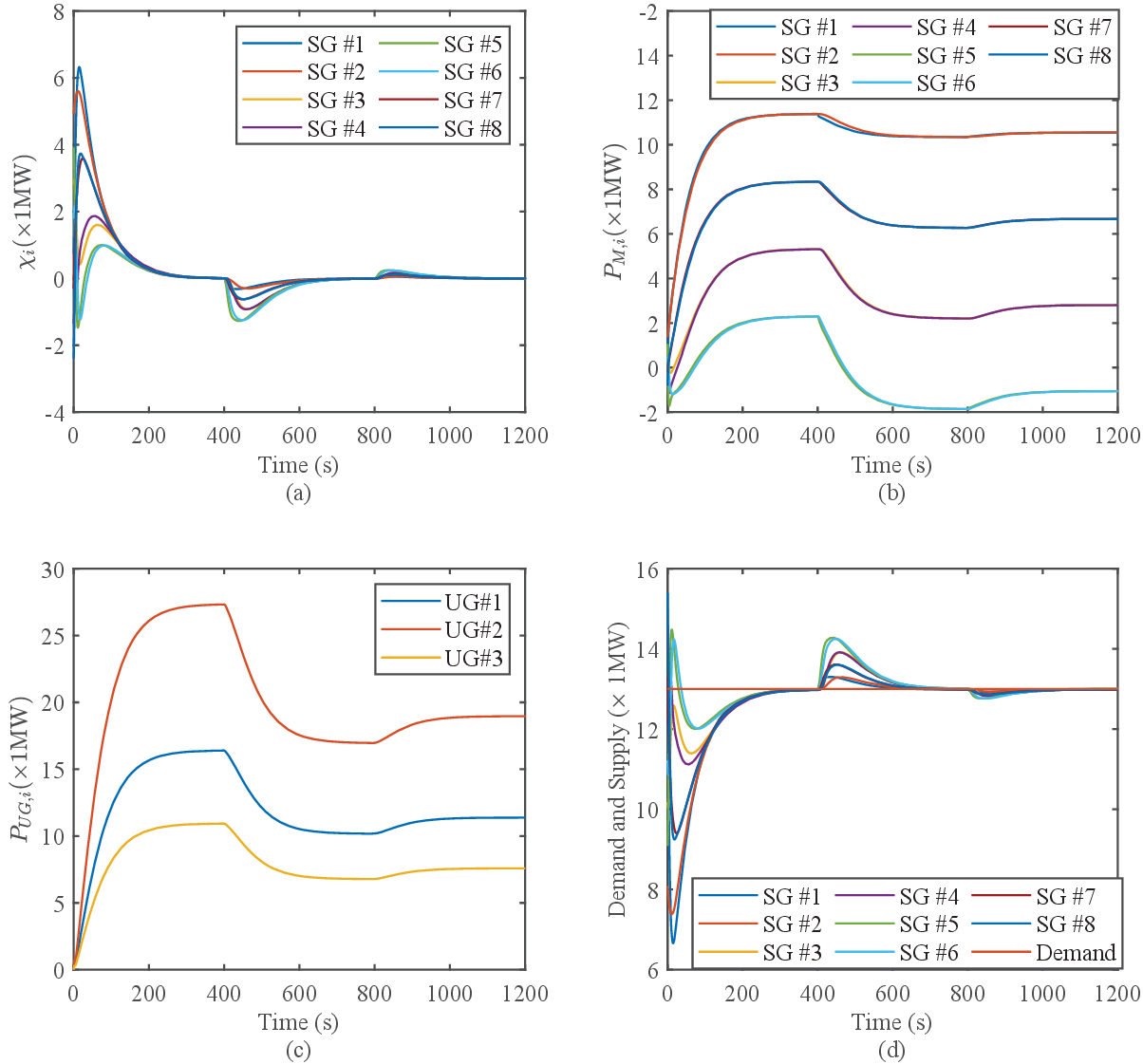}
	\caption{Evolution of total estimated power mismatch, exchange power, and supply-demand balance of each subgrid and the output power of each UG under tiered electricity price.}
	\label{UGJieti}
\end{figure}
\begin{figure}
	\centering
	\includegraphics[width=8cm]{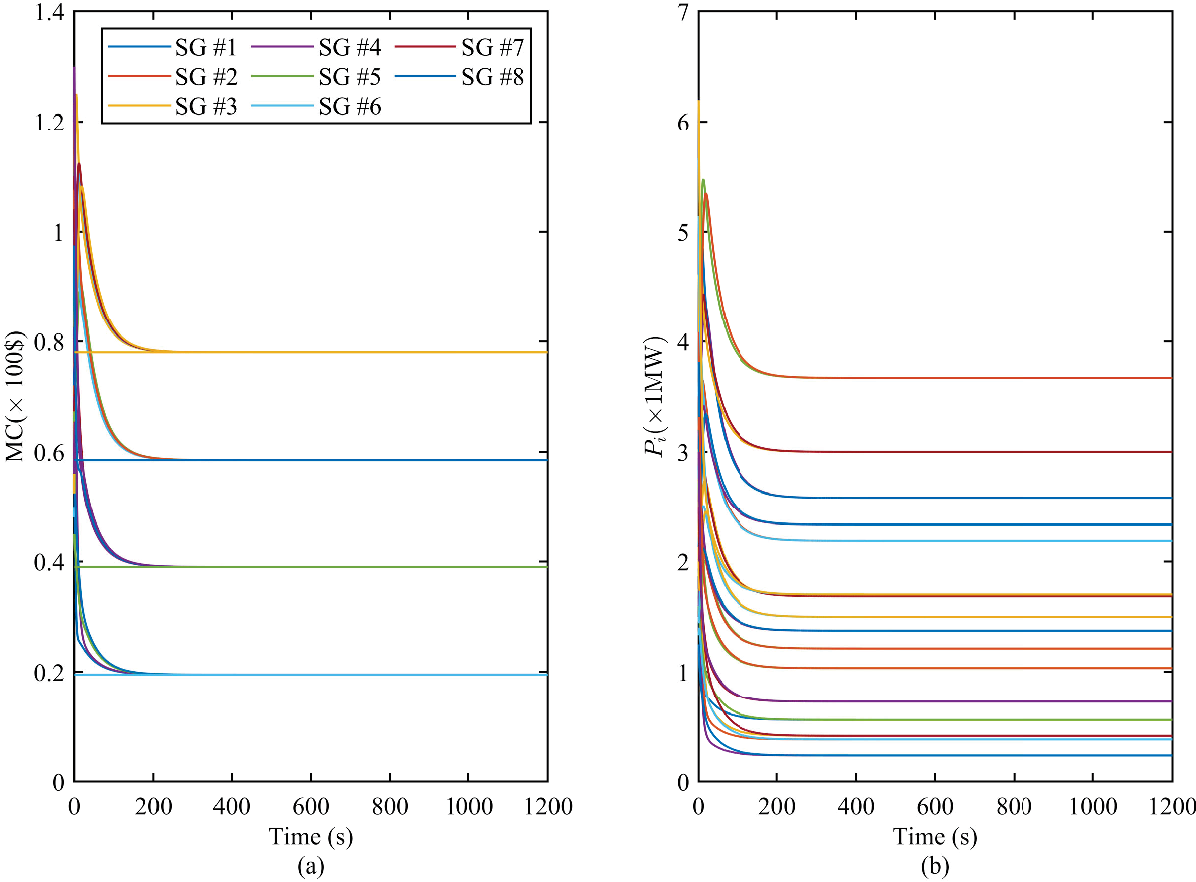}
	\caption{Evolution of MC and output power of each DG under load switching.}
	\label{demandMC}
\end{figure}
\begin{figure}
	\centering
	\includegraphics[width=8cm]{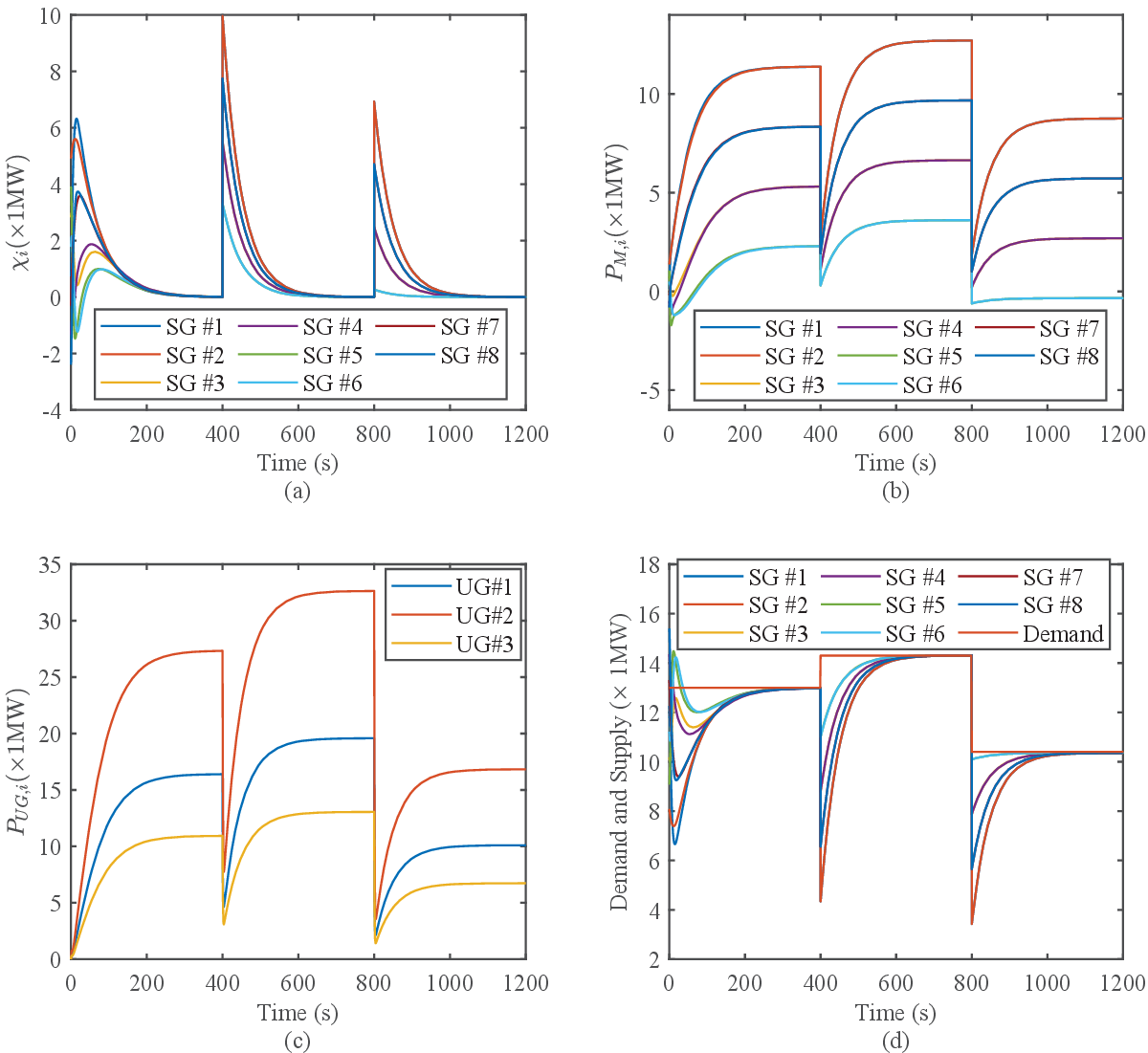}
	\caption{Evolution of total estimated power mismatch, exchange power, and supply-demand balance of each subgrid and the output power of each UG under load switching.}
	\label{demandUG}
\end{figure}
\begin{figure}
	\centering
	\includegraphics[width=8cm]{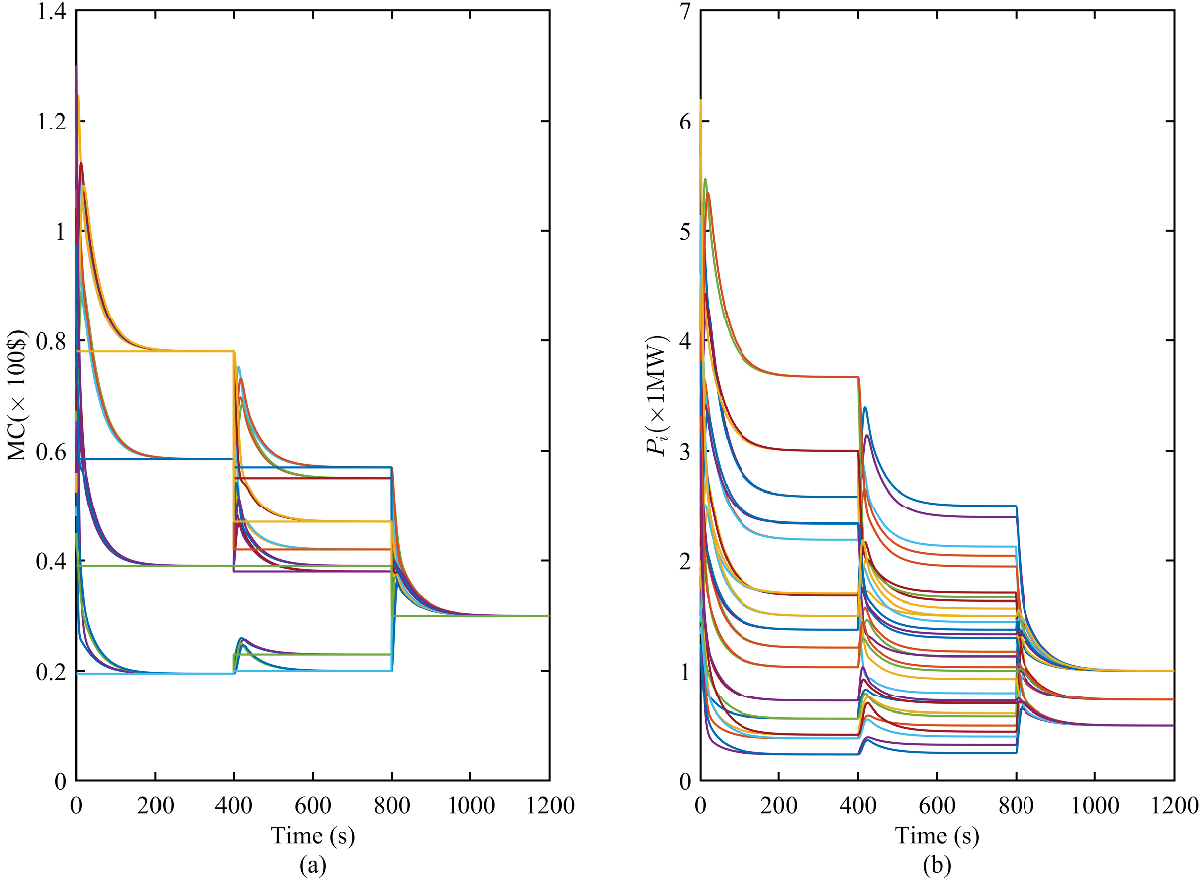}
	\caption{Evolution of MC and output power of each DG under changing aggregation coefficients.}
	\label{MCw}
\end{figure}
\begin{figure}
	\centering
	\includegraphics[width=8cm]{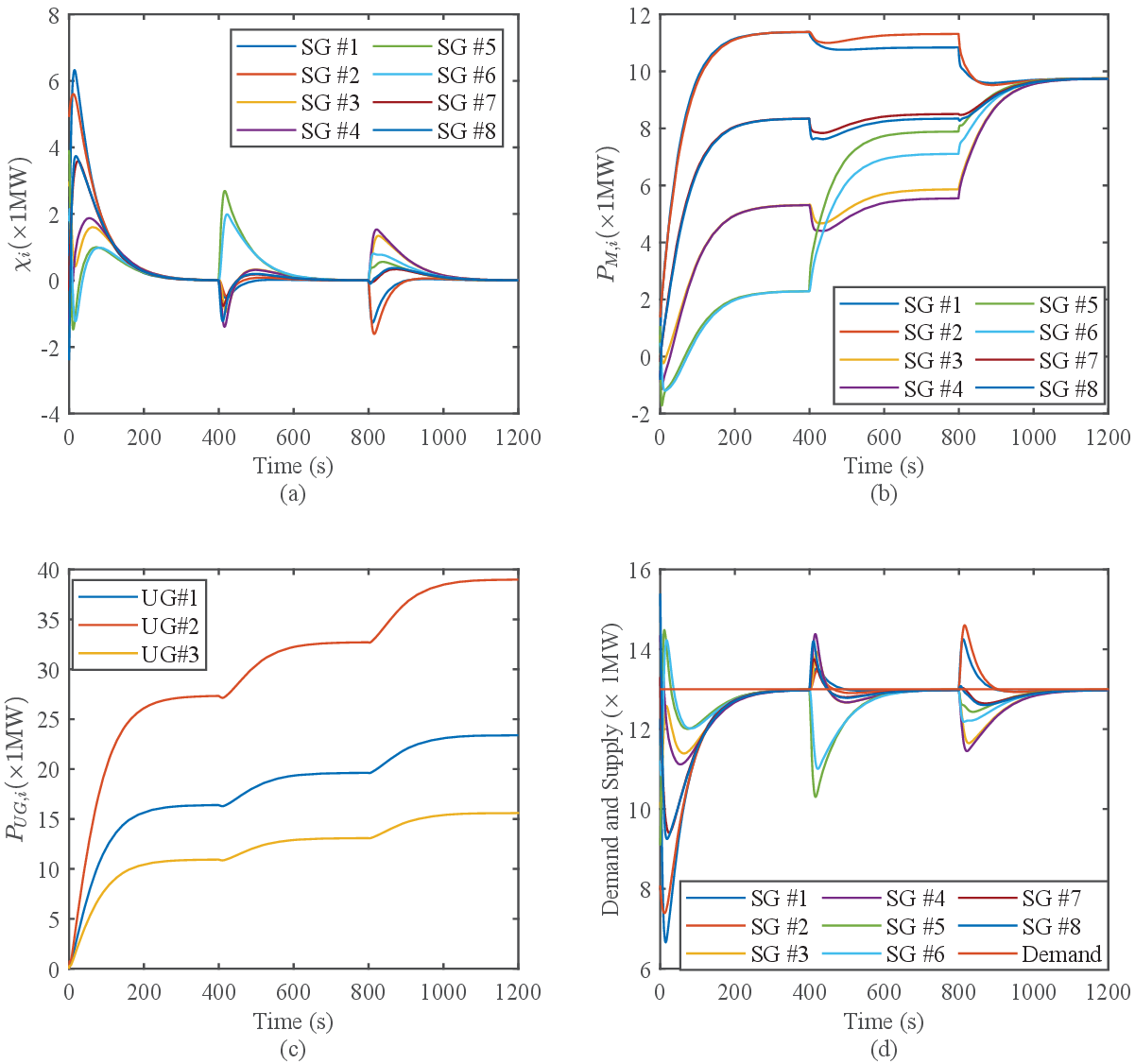}
	\caption{Evolution of total estimated power mismatch, exchange power, and supply-demand balance of each subgrid and the output power of each UG under changing aggregation coefficients.}
	\label{UGw}
\end{figure}
\par From Fig. \ref{MC_P}, it can be seen that although the communication network between flexible agents is strongly connected, by applying the designed schemes, MCs of each group of DGs can converge to the desired value as shown in Fig. \ref{MC_P}(a), resulting in stable output power as shown in Fig. \ref{MC_P}(b). From Fig. \ref{UG}, it can be seen that the estimated values of the average power mismatch for each group ultimately converge to 0 as shown in Fig. \ref{UG}(a). The power provided by VPP for each subgrid can ultimately ensure stability as shown in Fig. \ref{UG}(c). The communication weights calculated by Lemma \ref{lemwt} can result in the precise allocation of VPP power to UGs and ensure a balance between supply and demand of power in each subgrid as shown in Fig. \ref{UG}(d).
\subsection{Case 2. Simulation results under tiered electricity pricing}
\par The simulation in this case is similar to Case 1. In contrast, a tiered pricing strategy is applied here to test the designed solution, where time-phased electric price vectors are $\rho_{0\le t<400s}=[0.1,0.4,0.8]^\mathrm{T}$, $\rho_{400\le t<800s}=[0.2,0.4,0.6]^\mathrm{T}$, and $\rho_{800\le t\le1200s}=[0.1,0.6,0.8]^\mathrm{T}$. At this point, the communication weights between clusters are the same as in Case 1. Then, simulation results are shown in Figs. \ref{MC_PJieti} and \ref{UGJieti}.
\par From the simulation results, it can be seen that under the same communication weight, MC can converge to a tiered electricity price as shown in Fig. \ref{MC_PJieti}(a), resulting in stable output power as shown in Fig. \ref{MC_PJieti}(b). Meanwhile, the power mismatch can be estimated well and allocated among UGs as expected, ensuring a balance between supply and demand of power for each subgrid, as shown in Fig. \ref{UGJieti}(a), (c), and (d), respectively.
\subsection{Case 3. Simulation results under load switching}
\par Based on Case 1, this case is arranged to investigate the impact of the designed solution on load switching. Therein, there is an increase in each load node at $t=400s$ and a decrease in each load node at $t=800s$. The simulation results are shown in Figs. \ref{demandMC} and \ref{demandUG}.
\par Compared to Case 1, from the simulation results in Fig. \ref{demandMC}, the consensus states of MCs as shown in Fig. \ref{demandMC}(a), as well as the stability of the output power of each DG as shown in Fig. \ref{demandMC}(b), are not be affected by the change in load power. Although the average power mismatch of each subgrid fluctuates due to changes in load power, it eventually converges to 0 as shown in Fig. \ref{demandUG}(a). The power deficit in each subgrid can still be well estimated as shown in Fig. \ref{demandMC}(b) and successfully allocated to all UGs according to the predetermined ratio as shown in Fig. \ref{demandMC}(c), ensuring a balance between power supply and demand in each subgrid as shown in Fig. \ref{demandMC}(d).
\subsection{Case 4. Simulation results under changing aggregation coefficients}
\par VPP will adjust the power demand ratio for UGs based on the requirements for the power quality of a subgrid. Therefore, this case is arranged to investigate the performance of the designed solution in response to changes in aggregation coefficients. The weight matrix used in the aggregation strategy are shown in \eqref{WW}. The electricity price vector of UGs is $\rho=[0.1,0.4,0.8]^\mathrm{T}$. Thus, the aggregate electricity price vector for different time periods of all subgrids are
$W\rho=[0.195,0.195,0.585,0.585,0.78,0.78,0.39,0.39]^\mathrm{T}$ ($0<t\le 400s$), $W\rho=[0.23,0.2,0.55,0.57,0.42,0.47,0.38,0.39]^\mathrm{T}$ ($0<t\le 400s$), $W\rho=[0.3,0.3,0.3,0.3,0.3,0.3,0.3,0.3]^\mathrm{T}$ ($800\le t\le 1200s$). By Lemma \ref{lemwt}, the communication weight matrices used for each time period are shown in \eqref{Wt1} ($0\le t<400s$), \eqref{Wt2} ($400s\le t<800s$), and \eqref{Wt3} ($800s\le t\le1200s$). As a result, simulation results are shown in Figs. \ref{MCw} and \ref{UGw}.
\begin{figure*}[h]
	\centering
	\begin{equation}
		\label{WW}
		\begin{aligned}
			W_{0\le t<400s}=\left[\begin{matrix}
				0.71&0.27&0.02\\
				0.79&0.13&0.08\\
				0.05&0.45&0.5\\
				0.09&0.38&0.53\\
				0.016&0.022&0.962\\
				0&0.05&0.95\\
				0.18&0.71&0.11\\
				0.3&0.5&0.2
			\end{matrix}\right],W_{400s\le t<800s}=\left[\begin{matrix}
				0.7&0.2&0.1\\
				0.8&0.1&0.1\\
				0.1&0.45&0.45\\
				0.1&0.4&0.5\\
				0.2&0.6&0.2\\
				0.3&0.3&0.4\\
				0.2&0.7&0.1\\
				0.3&0.5&0.2
			\end{matrix}\right],
			W_{800s\le t\le 1200s}=\left[\begin{matrix}
				0.6&0.2&0.2\\
				0.6&0.2&0.2\\
				0.6&0.2&0.2\\
				0.6&0.2&0.2\\
				0.6&0.2&0.2\\
				0.6&0.2&0.2\\
				0.6&0.2&0.2\\
				0.6&0.2&0.2
			\end{matrix}\right].
		\end{aligned}
	\end{equation}
	\begin{equation}
			\label{Wt2}
			\Omega_{400s\le t<800s}=\left[\begin{matrix}
				0&0&0.5833&0&0&0&0&0\\
				0.1698&0&0&0.421&0&0&0&0\\
				0&0&0&0.9323&0&0&0.6214&0\\
				0&0&0&0&0&0&0&2.0393\\
				2.5486&0&0&0&0&0&0&0\\
				0&0.6059&0&0&1.2728&0&0&0\\
				0&0&0&0&0.6778&0&0&0.6294\\
				0&0&0&0&0&1.158&0&0
			\end{matrix}\right].
		\end{equation}
		\begin{equation}
			\label{Wt3}
			\Omega_{800s\le t\le 1200s}=\left[\begin{matrix}
				0&0&1.3953&0&0&0&0&0\\
				0.6977&0&0&0.6977&0&0&0&0\\
				0&0&0&0.6977&0&0&0.6977&0\\
				0&0&0&0&0&0&0&1.3953\\
				1.3953&0&0&0&0&0&0&0\\
				0&0.6977&0&0&0.6977&0&0&0\\
				0&0&0&0&0.6977&0&0&0.6977\\
				0&0&0&0&0&1.3953&0&0
			\end{matrix}\right].
		\end{equation}
	\end{figure*}
\par The change in aggregation coefficients results in different communication weights. Driven by schemes \eqref{lam1}, \eqref{ci}, \eqref{pm2}, \eqref{ER1}, \eqref{agg}, and \eqref{ci2}, all MCs of DGs can achieve the expected leader-follower cluster consensus as shown in Fig. \ref{MCw}(a), resulting in stable output power as shown in Fig. \ref{MCw}(b). Similarly, the average power mismatch estimation converges to 0 in each subgrid as shown in Fig. \ref{UGw}(a), and the power deficit of each subgrid can be successfully estimated as shown in Fig. \ref{UGw}(b). Based on the above, the power deficit of each subgrid is evenly distributed among UGs according to a predetermined ratio as shown in Fig. \ref{UGw}(c). Meanwhile, the supply and demand balance of power in each subgrid can be maintained well as shown in Fig. \ref{UGw}(d).
\begin{table}[h]\label{tab0}
	\caption{Comparison on algorithms with previous methods}
	\centering
	\begin{tabular}{cccc}
		\toprule
		Schemes &\makecell[c]{The one in\\this paper}&\makecell[c]{The one in\\ \cite{10302354}}&\makecell[c]{The one in\\ \cite{zaeryDistributedGlobalEconomical2020}}\\
		\midrule
		Consensus control & $\checkmark$ & $\checkmark$ & $\checkmark$\\
		Multiconsensus control & $\checkmark$ &$\times$ & $\times$\\
		Containment control& $\checkmark$ & $\times$ & $\times$\\
		Prescribed steady state&$\checkmark$ & $\times$ & $\times$\\
		\bottomrule
	\end{tabular}
\end{table}
\begin{table}[h]\label{tab1}
	\caption{Comparison on operating conditions with previous methods}
	\centering
	\begin{tabular}{cccc}
		\toprule
		Schemes &\makecell[c]{The one in\\this paper}&\makecell[c]{The one in\\ \cite{10302354}}&\makecell[c]{The one in\\ \cite{zaeryDistributedGlobalEconomical2020}}\\
		\midrule
		Multiple UGs & $\checkmark$ & $\times$ & $\times$\\
		Multiple subgrids & $\checkmark$ &$\times$ & $\times$\\
		\makecell[c]{Differential demand\\between subgrids}& $\checkmark$ & $\times$ & $\times$\\
		\makecell[c]{Switching communication\\topology}&$\times$ & $\checkmark$ & $\checkmark$\\
		\bottomrule
	\end{tabular}
\end{table}
\section{Conclusion}
In this article, a distributed cluster ED scheme induced by communication topology is investigated. The MCs in each subgrid converge to the same electricity price, and the MCs in different subgrids converge to different electricity prices, where the electricity price in each subgrid is a convex combination of several UG electricity prices. Based on VPP, the power deficit in each subgrid is well estimated and proportionally allocated to UGs according to predetermined convex combination coefficients. It is worth mentioning that the communication topology between clusters used is directed and connected, and the controller is designed by calculating the required communication weights to achieve the goal. The simulation results show that the designed distributed scheme can achieve cluster ED and ensure power supply and demand balance under load switching, tiered pricing, and variable aggregation coefficient. Compared with previous schemes, the designed scheme has significant advantages in algorithm design and application scenarios, as shown in Tables I and II. However, the current scheme is still a simplified version, and physical constraints such as prohibited operating zones of generators, spinning reserve requirements and max/min power range have not been included in the designed scheme. In view of this, the designed scheme will be improved in the future to enhance practicality. In addition, in the face of frequent uncertainties in communication networks, such as network attacks and communication failures, the solution proposed in this article should be further developed to improve the reliability of communication networks.
\appendices
\ifCLASSOPTIONcaptionsoff
\newpage
\fi
\bibliographystyle{IEEEtran}
\bibliography{BESS_reference.bib}
\end{document}